\documentclass[conference]{lmcs}
\pdfoutput=1
\usepackage[utf8]{inputenc}

\usepackage{lastpage}
\lmcsdoi{22}{2}{3}
\lmcsheading{}{\pageref{LastPage}}{}{}%
{Jun.~02,~2025}{Apr.~09,~2026}{}

\usepackage{cite}
\usepackage{hyperref}
\hypersetup{
  colorlinks=true,    
  linkcolor=blue,     
  citecolor=blue,     
  urlcolor=blue       
  allcolors=blue
}

\usepackage{amsmath,amssymb,amsfonts,amsthm}
\usepackage{algorithmic}
\usepackage{graphicx}
 \usepackage{relsize}
\usepackage{textcomp}
\usepackage{tikz}
\newtheorem{theorem}[thm]{Theorem}
\newtheorem{definition}[thm]{Definition}
\newtheorem{observation}[thm]{Observation}

\newtheorem{lemma}[thm]{Lemma}
\newtheorem{example}[thm]{Example}
\newtheorem{corollary}[thm]{Corollary}

\usepackage{mathtools}

\usepackage{prettyref}%
\newrefformat{sec}{Section~\ref{#1}}%
\newrefformat{def}{Definition~\ref{#1}}%
\newrefformat{thm}{Theorem~\ref{#1}}%
\newrefformat{fig}{Figure~\ref{#1}}%
\newrefformat{tab}{Table~\ref{#1}}%
\newrefformat{alg}{Algorithm~\ref{#1}}%
\newrefformat{fun}{Function~\ref{#1}}%
\newrefformat{lst}{Listing~\ref{#1}}%
\newrefformat{appendix}{Appendix~\ref{#1}}%
\newrefformat{fm}{Formula~\ref{#1}}%
\newrefformat{ex}{Example~\ref{#1}}%

\newcommand{\bigpl}{\mathlarger{(}}
\newcommand{\bbigpl}{\mathlarger{\mathlarger{(}}}
\newcommand{\bbbigpl}{\mathlarger{\mathlarger{\mathlarger{(}}}}
\newcommand{\bbbbigpl}{\mathlarger{\mathlarger{\mathlarger{\mathlarger{(}}}}}
\newcommand{\bigpr}{\mathlarger{)}}
\newcommand{\bbigpr}{\mathlarger{\mathlarger{)}}}
\newcommand{\bbbigpr}{\mathlarger{\mathlarger{\mathlarger{)}}}}
\newcommand{\bbbbigpr}{\mathlarger{\mathlarger{\mathlarger{\mathlarger{)}}}}}

\usepackage{soul}
\newcommand{\unif}{\ensuremath{\stackrel{{\scriptscriptstyle ?}}{=}}}
\newcommand{\Vars}{\ensuremath{\mathcal{V}}}
\newcommand{\Varsf}{\ensuremath{\mathcal{V}_f}}
\newcommand{\Varsi}{\ensuremath{\mathcal{V}_i}}

\newcommand{\dep}[1]{\ensuremath{\mathit{dep}(#1)}}
\newcommand{\arity}[1]{\ensuremath{\mathit{arity}(#1)}}

\newcommand{\sub}[1]{\ensuremath{\mathit{sub}(#1)}}

\newcommand{\mulOp}[1]{\ensuremath{\mathtt{mul}\mathtt{[} #1\mathtt{]}}}

\newcommand{\mulrOp}[1]{\ensuremath{\mathtt{mul}\mathtt{[} #1\mathtt{]}_{r}}}

\newcommand{\mullOp}[1]{\ensuremath{\mathtt{mul}\mathtt{[} #1\mathtt{]}_{l}}}

\newcommand{\cntOp}[1]{\ensuremath{\mathtt{cnt}\mathtt{[} #1\mathtt{]}}}

\newcommand{\cntrOp}[1]{\ensuremath{\mathtt{cnt}\mathtt{[} #1\mathtt{]}_{r}}}
\newcommand{\cntlOp}[1]{\ensuremath{\mathtt{cnt}\mathtt{[} #1\mathtt{]}_{l}}}

\newcommand{\cvtpOp}[1]{\ensuremath{\mathtt{cvt}\mathtt{[} #1\mathtt{]}^+}}
\newcommand{\cvtnOp}[1]{\ensuremath{\mathtt{cvt}\mathtt{[} #1\mathtt{]}^-}}

\newcommand{\cvtsOp}[1]{\ensuremath{\mathtt{cvt}\mathtt{[} #1\mathtt{]}^\star}}

\newcommand{\proofcase}[1]{\par\noindent\textbf{#1}}

\def\BibTeX{{\rm B\kern-.05em{\sc i\kern-.025em b}\kern-.08em
    T\kern-.1667em\lower.7ex\hbox{E}\kern-.125emX}}

\allowdisplaybreaks

\begin{document}

\title[One is all you need]{One is all you need: Associative Second-order Unification without First-order Variables}
  
\thanks{Funded by Czech Science Foundation Grant No. 22-06414L and Cost Action CA20111 EuroProofNet and the Austrian Science Fund (FWF) project AUTOSARD (36623), the European Research Council (ERC) project FormalWeb3 (Grant ID 101156734) and LASD (Grant ID: 101089343).}	

\author[D.M.~Cerna]{David M. Cerna\lmcsorcid{0000-0002-6352-603X}}[a,b]
\author[J.~Parsert]{Julian Parsert\lmcsorcid{0000-0002-5113-0767}}[c]

\address{Institute of Computer Science, Czech Academy of Sciences}	
\email{dcerna@cs.cas.cz}  
\address{Dynatrace Research}	
\email{david.cerna@dynatrace.com}  
\address{RPTU Kaiserslautern}	
\email{julian.parsert@gmail.com}  

\begin{abstract}
  We introduce a fragment of second-order unification, referred to as
  \emph{Second-Order Ground Unification (SOGU)}, with the following properties:
  (i) only one second-order variable is allowed and (ii) first-order variables
  do not occur. We study an equational variant of SOGU where the signature
  contains \textit{associative} binary function symbols (ASOGU) and show that
  Hilbert's 10$^{th}$ problem is reducible to ASOGU unifiability, thus proving
  undecidability. Our reduction provides a deeper understanding of the decidability boundary for 
(equational) second-order unification, as previous results required
  first-order variable occurrences, multiple second-order variables, and/or
  equational theories involving \textit{length-reducing} rewrite
  systems. Furthermore, our reduction holds even in the case when associativity
  of the binary function symbol is restricted to \emph{power associative},
  i.e. f(f(x,x),x)= f(x,f(x,x)).
\end{abstract}

\maketitle

\section{Introduction}\label{sect:introduction}
In general, unification is the process of equating symbolic
expressions. Second-order unification concerns symbolic expressions containing function variables, i.e., variables that take expressions as arguments. Such processes are fundamental to mathematics and computer science and are central to formal methods, verification, automated reasoning, interactive theorem proving, and various other areas. In addition, formal verification methods based on \emph{satisfiability modulo theories (SMT)} exploit various forms of
unification within the underlying theories and their implementations.

In this paper, we reduce Hilbert's 10$^{th}$ problem to a fragment of second-order unification restricted as follows: (i) only one second-order variable occurs, (ii) no first-order variables occur, and (iii) an associative binary function symbol is allowed to occur. 

Essentially, our encoding maps unknowns of the given polynomial to argument positions of the unique second-order variable occurring in the unification problem. Replacement of an unknown by a non-negative integer $n$ is equivalent to applying a substitution to the unification problem, which replaces \textbf{F} by a term containing $n$ occurrences of the associated bound variable.

Such unification problems are related to recent investigations aiming to increase the expressive power of SMT
by adding higher-order features~\cite{DBLP:conf/cade/BarbosaROTB19,Tourret24}. Some
methods for finding SMT models use synthesis
techniques~\cite{DBLP:conf/lpar/ParsertBJK23} such as \textit{Syntax-Guided Synthesis}
(SyGuS)~\cite{DBLP:series/natosec/AlurBDF0JKMMRSSSSTU15}, a common approach for
function synthesis problems. In some cases, SyGuS can be considered a form of
equational second-order unification where (i) only one second-order variable is
allowed and (ii) first-order variables do not occur. Often, enumerative SyGuS
solvers use \textit{Counterexample Guided Inductive
Synthesis} (CEGIS) ~\cite{DBLP:conf/cav/AbateDKKP18} to speed up the synthesis procedure
by leveraging ground instances of the problem. In the synthesis domain of
\textit{Programming-By-Example} (PBE), the goal is to find functions that satisfy a given
set of concrete input-output examples where no variables (other than the
synthesis target) are present~\cite{DBLP:journals/ftpl/GulwaniPS17}. Combining
these developments motivates the investigation of second-order unification,
including \emph{ground} cases without first-order variables. 

Already in the 1970s, Huet and Lucchesi proved the undecidability of
higher-order
logic~\cite{DBLP:journals/iandc/Huet73,lucchesi1972undecidability}. Concerning
the equational variant of higher-order
unification~\cite{Snyder90,NipkowQ91,DoughertyJ92}, undecidability follows from
Huet's result. Early investigations~\cite{Goubault94} discovered decidable
fragments of higher-order E-unification. However, these are not interesting to
the work presented in this paper. Goldfarb~\cite{DBLP:journals/tcs/Goldfarb81}
strengthened Huet's undecidability result by proving second-order unification
undecidable. Both results only concern the general unification problem, thus
motivating the search for decidable fragments and honing the undecidability
results (See~\cite{DBLP:conf/unif/Levy14} for a comprehensive survey). Known
decidable classes include \emph{Monadic
  Second-order}~\cite{DBLP:journals/apal/Farmer88,DBLP:journals/siamcomp/LevySV08},
\emph{Second-order Linear}~\cite{DBLP:conf/rta/Levy96}, \emph{Bounded
  Second-order}~\cite{DBLP:journals/iandc/Schmidt-Schauss04}, and \emph{Context
  Unification}~\cite{DBLP:conf/icalp/Jez14}. Concerning second-order
E-unification, the authors of~\cite{OTTO19981} present several decidable and
undecidable fragments using a reduction from word unification problems and
length-reducing equational theories. Undecidability of Second-order unification
has been shown for the following fragments:
\begin{itemize}
\item two second-order variables, no first-order
  variables,~\cite{DBLP:journals/iandc/LevyV00}
\item one second-order variable \emph{with at least eight} first-order
  variables~\cite{DBLP:conf/asian/GanzingerJV98},
\item one second-order variable with only ground arguments and first-order
  variables~\cite{DBLP:journals/iandc/LevyV00}, and
\item two second-order variables over a \textit{length-reducing} equational
  theory~\cite{OTTO19981}.
\end{itemize}
Interestingly, Levy~\cite{DBLP:conf/unif/Levy14} notes that the number of
second-order variables only play a minor role in the decidability of fragments, as Levy and
Veanes~\cite{DBLP:journals/iandc/LevyV00} provide a reduction translating
arbitrary second-order equations to equations containing only one second-order
variable and \emph{additional first-order variables}. These results immediately
lead to the following question:
\begin{quote}
  How important are first-order variables for the undecidability of second-order
  unification?
\end{quote}
To address this question, we investigate (\emph{associative}) \emph{second-order ground unification} where only one second-order variable (arbitrary occurrences) is allowed, \emph{no first-order variables} occur, and some binary function symbols are interpreted as associative. To this end, we introduce two functions related to the \emph{multiplicity operator}~\cite{DBLP:journals/jsyml/Vrijer85,DBLP:conf/csl/AccattoliK10}, the $n$-counter and the $n$-multiplier, that allow us to reason about the multiplicity of function symbols with respect to a given substitution and of function symbols introduced by the substitution, respectively. These functions allow us to describe the properties of the unification problem number theoretically. As a result, we can reduce finding solutions to \emph{Diophantine equations} to a \emph{unification condition} involving the structure of the substitution and $n$-counter, thus proving undecidability. In particular, our contributions are as follows:
\begin{itemize}
\item We introduce the $n$-counter and $n$-multiplier and prove essential
  properties of both.
\item We prove the undecidability of associative second-order unification with
  \emph{one} function variable and \emph{no} first-order variables by showing undecidability over a restricted signature: a single constant and a binary function symbol that is interpreted as power associative, i.e., f(x,f(x,x))=f(f(x,x),x). 
\end{itemize}
To illustrate our encoding, let us consider a few simple examples. First, consider the polynomial $p(x) = x-2$. We encode this polynomial as the following \textit{unification problem} where
$\textbf{F}$ is a second-order variable, $a$ is a constant, $g$ is an
associative binary function symbol:
\begin{gather*}
  \textbf{F}(g(a,a)) \unif g(a,a,\textbf{F}(a))
\end{gather*}
Observe that we \textit{flattened} nested associative function symbols, i.e., intermediary occurrences of
associative function symbols are dropped. Furthermore, the arity of $\textbf{F}$
is equivalent to the number of unknowns in $p(x)$. The consecutive \textbf{a}'s in the term  $g(a,a,\textbf{F}(a))$ denote the $2$ in $p(x)$. The minus sign is denoted by the consecutive \textbf{a}'s occurring on the right-hand side of the unification  problem and the absence of any prefix on the left-hand side.
The consecutive \textbf{a}'s in the term $\textbf{F}(g(a,a))$ denote the coefficient
of $x$ in $p(x)$ (i.e. 1), since subtracting the
number of occurrences of \textbf{a}'s in $\textbf{F}(a)$ from occurrences in $\textbf{F}(g(a,a))$ results in 1. 

The obvious solution to the equation $p(x) = 0$ is $x=2$. Observe that the substitution $ \{F\mapsto \lambda y.g(y,y)\}$ where the bound variable occurs twice is a solution to the unification problem we derived from $p(x)$. In our encoding, the number of occurrences of a bound variable in the solution is precisely the value that should be substituted into the associated unknown. 
\begin{gather*}
\textbf{F}(g(a,a))\{F\mapsto \lambda y.g(y,y)\} = (\lambda y.g(y,y))g(a,a) \rightarrow_{\beta} g(a,a,a,a)\\
g(a,a,\textbf{F}(a))\{F\mapsto \lambda y.g(y,y)\} = g(a,a,(\lambda y.g(y,y))a)\rightarrow_{\beta} g(a,a,a,a)
\end{gather*}
Let us consider a slightly more complex example: $p(x) = x^2 -4x + 3$. Note that we will only consider natural number solutions to the polynomial; thus, we have selected a quadratic with two natural number solutions for illustration. We encode this polynomial as the following \textit{unification problem} where
$\textbf{F}$ is a second-order variable, $a$ is a constant, $g$ is an
associative binary function symbol:
\begin{gather*}
  g(a,a,a,\textbf{F}(\textbf{F}(g(a,a))) \unif \textbf{F}(g(a,a,a,a,\textbf{F}(a)))
\end{gather*}

Observe that we now have nested function variables as the degree of $x$ is 2 in one of the monomials. Furthermore, consider the arguments to the outermost occurence of $\textbf{F}$, $\textbf{F}(g(a,a))$ and $g(a,a,a,a,\textbf{F}(a))$. Together, these terms form the unification problem for the polynomial $x-4$. When dealing with polynomials of higher degree, we group the monomials together by variable, reduce their degree by 1, and recursively call the procedure. 
 
It is easily derived that $p(x)$ has two solutions, namely, $x=3$ and $x=1$. From these solutions, we can derive substitutions for the unification problem with the appropriate number of bound variables, i.e., 
$ \{F\mapsto \lambda y.y\}$ and $ \{F\mapsto \lambda y.g(y,y,y)\}$. Observe the following: 
\begin{align*}
g(a,a,a,\textbf{F}(\textbf{F}(g(a,a)))\{F\mapsto \lambda y.y\}\ =&\ g(a,a,a,(\lambda y.y)((\lambda y.y)g(a,a)) \rightarrow_{\beta} g(a,a,a,a,a)\\
\textbf{F}(g(a,a,a,a,\textbf{F}(a)))\{F\mapsto \lambda y.y\}\ =&\ (\lambda y.y)(g(a,a,a,a,((\lambda y.y))a)\rightarrow_{\beta} g(a,a,a,a,a)
\end{align*}
We leave it as an exercise to verify the other solution, as it yields a term containing 21 occurrences of $a$. To illustrate how the number of unknowns influences the function variable, let us consider the polynomial $p(x,y) = xy-3x-2y+6$ or $(x-2)(y-3)$ when factored. The result unification problem is as follows (See Figure~\ref{fig:multiVarExample} for a graphical illustration):

\begin{gather*}
  g(a^6,\textbf{F}(\textbf{F}(a,g(a,a)),a)) \unif \textbf{F}(g(a,a,a,\textbf{F}(a,a)),g(a,a,a))
\end{gather*}
\begin{figure*}
  \centering
  \begin{tikzpicture}[ level 1/.style={level distance = 1.5cm, sibling distance
      = 1.5cm}, level 2/.style={level distance = 1.5cm, sibling distance =
      1.8cm}, level 3/.style={level distance = 1.5cm, sibling distance = 1cm},
    level 4/.style={level distance = 1.5cm, sibling distance = 1cm}, ]
    \node[xshift=-4cm]{$x\cdot y-3\cdot x-2\cdot y+6$} child {node {$6$}} child
    {node {$y-3$} child {node {$-3$}  } child {node
        {$0$}} child {node
        {$1$}} } child {node {$-2$}  };

    \node[xshift=6cm,yshift=1.5cm]{} child{ node {$F$} child {node {$g$} child
    {node {$a,a,a$}} child {node {$F$}  child {node
          {$a$}}child {node {$a$}} } } child {node {$g(a,a,a)$}} edge from parent[draw=none]};
      
     \node[xshift=0cm]{$g$} child {node {$a^6$}} child {node {$F$} child {node
        {$F$} child {node {$a$}} child {node {$g(a,a)$}} } child {node {$a$}} };

    \draw (-4,-5.5) node {$p(x,y)$}; \draw (0,-5.5) node {Left Side};
    \draw (6,-5.5) node {Right Side}; 
\draw[red,thick]  (-4,-1.5) circle (.47cm);
\draw[red,thick]  (.75,-1.5) circle (.4cm);
\draw[red,thick]  (5.1,-1.5) circle (.38cm);

  \end{tikzpicture}
  \caption{Construction of terms from a polynomial. The red circles denote the start of terms representing the encircled polynomial.}
  \label{fig:multiVarExample}
\end{figure*}
where $a^n$ abbreviates a sequence of $n$ occurrences of $a$. Observe that, within the outermost occurrence of \textbf{F}, each position contains a representation of a recursively derived polynomial. Furthermore, the groupings are not independent. Let use consider the nested terms in the $x$ position, $\textbf{F}(a,g(a,a))$ and $g(a,a,a,\textbf{F}(a,a))$ denoting the polynomial $y-3$. Even though $y-3$ does not contain the unknown $x$, we must include the $x$ position as the initial problem has two unknowns. We put a single \textbf{a} in the position of the second-order variable to denote no information. Now consider the  nested problem in the $y$ position, $a$ and $g(a,a,a)$ denoting the polynomial $-2$. Observe that even though $y$ occurs in both $xy$ and $-2y$, we cannot consider the $xy$ monomial as it was "used up" by the first problem. Thus, we only need to consider the representation of $-2$. This is the interdependence mentioned above. 

From the solution  $x=2$ and $y=1$, we can build the substitution $ \{F\mapsto \lambda y,z.g(y,y,z)\}$ that, when applied to the terms in the unification problem, results in a term with 15 occurrences of $a$. A more complex example is presented in Section~\ref{app:undecidable} after we introduce the encoding and prove important properties of it.

Observe that through our encoding, any decidable class of Diophantine equations
provides a decidable fragment of the second-order unification problem presented
in this work. Furthermore, our reduction uses a simple encoding that guarantees
the equation presented in Lemma~\ref{lem:MainEq} directly reduces to
$0=p(\overline{x_n})$ where $p(\overline{x_n})$ denotes a polynomial with
integer coefficients over the variables $x_1,\cdots,x_n$.

There are likely more intricate encodings based on the $n$-counter and the $n$-multiplier, which map polynomials to more interesting unification problems. 
Furthermore, it remains open whether second-order ground unification is decidable, i.e., without function symbols interpreted as \textit{power associative}.

The rest of the paper is as follows:
\begin{itemize}
\item In section~\ref{sec:mulcnt}, we introduce the concepts of \emph{$n$-multiplier} (Definition~\ref{def:nmulti}) and \emph{$n$-counter} (Definition~\ref{def:ncounter}). These concepts provide a method for counting the number of occurrences of a given constant within a term after applying a substitution with a fixed number of occurrences of each bound variable. 
\item We go on to show properties of these functions, specifically, how these functions relate to one another (Lemma~\ref{lem:FixingConstants}) and how the functions relate to unifiability of two terms (Lemma~\ref{lem:MainEq}). 
\item In section~\ref{app:undecidable}, we first focus on functions for manipulating polynomials. First, we introduce \emph{monomial groups} (Definition~\ref{def:monogrp}), which group monomials by the unknowns occurring within them. Specifically, we order the unknowns and place the monomials into the group associated with the smallest unknown that occurs within it. This restructuring of the polynomials allows us to recursively descend the structure and build terms from it. 
\item The function recursively building terms from the polynomials is the \emph{$n$-converter} (Definition~\ref{def:convert}), which has two modes depending on which side of the unification problem the term comes from. 
\item Theorem~\ref{prop:muleqpoly} illustrates the importance of these terms, computing the counter of the terms and then taking the difference results in the input polynomial. We use this result to complete the reduction (Lemma~\ref{lem:reduction-sound-complete}).
\end{itemize}

\section{preliminaries}\label{sect:prelims}
Let $\Sigma$ be a \emph{signature} of function symbols with a fixed arity. For
$f\in \Sigma$, the arity of $f$ is denoted by $\mathit{arity}(f)\geq 0$ and if
$\mathit{arity}(f)=0$ we refer to $f$ as a \emph{constant} (For the rest of the paper $a$ and $b$ are the only distinct constant symbols and $f$ and $g$ are the only distinct function symbols).

By $\Vars$, we denote a countably infinite set of \emph{variables}. Furthermore,
let $\Varsi,\Varsf \subset \Vars$ such that $\Vars_i\cap\Vars_f=\emptyset$. We refer to members
of $\Varsi$ as \emph{individual variables}, denoted by $x,y,z$, $\ldots$ and members
of $\Varsf$ as \emph{function variables}, denoted by $F,G,H$, $\ldots$. Members of
$\Varsf$ have an arity $> 0$ which we denote by $\mathit{arity}(F)$ where
$F\in\Varsf$. By $\Varsf^n$, where $n>0$, we denote the set of all function
variables with arity $n$.

We refer to members of the term algebra $\mathcal{T}(\Sigma,\Vars)$ as \emph{terms}. By
$\Varsi(t)$ and $\Varsf(t)$ ($\Varsf^n(t)$ for $n\geq 0$), we denote the set of
individual variables and function variables (with arity $=n$) occurring in $t$,
respectively. We refer to a term $t$ as \emph{$n$-second-order ground} ($n$-SOG)
if $\Varsi(t)=\emptyset$, $\Varsf(t)\not=\emptyset$ with
$\Varsf(t)\subset\Varsf^n$, \emph{first-order} if $\Varsf(t)=\emptyset$, and \emph{ground} if
$t$ is first-order and $\Varsi(t)=\emptyset$.
When possible, without causing confusion, we will abbreviate a sequence of terms
$t_1,\ldots, t_n$ by $\overline{t_n}$ where $n>0$. The \emph{set of subterms of a
  term $t$} is denoted $\sub{t}$. The \emph{number of occurrences of a symbol
  (or variable) $f$ in a term $t$} is denoted $\mathit{occ}_{\Sigma}(f,t)$.

A \emph{$n$-second-order ground ($n$-SOG) unification equation} has the form
$u\unif_F v$ where $u$ and $v$ are $n$-SOG terms and $F\in \Vars_f^n$ such that
$\Vars_f(u)=\{F\}$ and $\Vars_f(v)=\{F\}$. A \emph{$n$-second-order ground
  unification problem} ($n$-SOGU problem) is a pair $(\mathcal{U},F)$ where
$\mathcal{U}$ is a set of $n$-SOG unification equations and $F\in \Vars_f^n$ such that for
all $u\unif_G v\in \mathcal{U}$, $G=F$. Recall from the definition of $n$-SOG that
$\Varsi(u) = \Varsi(v) = \emptyset$. When possible, without causing confusion, we will write SOG, SOGU, $\ldots$ (i.e., drop the $n$- prefix)

We define the depth of a term $t$, denoted $\dep{t}$ inductively as follows: (i)
if $t\in \Varsi$ or $\arity{t}=0$, then $\dep{t} =1$, (ii) $F\in \Vars_f$ and
$t= F(t_1,\cdots, t_n)$, then $\dep{t} =1+ \max\{\dep{t_i}\mid 1\leq i\leq n\}$, and (iii) if
$t= f(t_1,\cdots, t_n)$, then $\dep{t} =1+ \max\{\dep{t_i}\mid 1\leq i\leq n\}$

A \emph{substitution} is a set of bindings of the form
$\{F_1\mapsto \lambda \overline{y_{l_1}}.t_1,\ldots F_k\mapsto \lambda \overline{y_{l_k}}.t_k,x_1\mapsto s_1,\ldots,
x_w\mapsto s_w\}$ where $k,w\geq 0$, for all $1\leq i \leq k$, $t_i$ is first-order and
$\Vars_i(t_i)\subseteq \{y_1,\ldots,y_{l_i}\}$, $\mathit{arity}(F_i)=l_i$, and for all
$1\leq i \leq w$, $s_i$ is ground. Given a substitution $\sigma$,
$\mathit{dom}_f(\sigma) = \{ F\mid \{F\sigma = \lambda \overline{x_n}.t\in \sigma\wedge F\in \Vars_f^n\}$ and
$\mathit{dom}_i(\sigma) = \{ x\mid x\sigma \not = x \wedge x\in \Vars_i\}$. We refer to a
substitution $\sigma$ as second-order when $\mathit{dom}_i(\sigma) = \emptyset $ and first-order
when $\mathit{dom}_f(\sigma) = \emptyset $.  We use postfix notation for substitution
applications, writing $t\sigma$ instead of $\sigma(t)$. Substitutions are denoted by
lowercase Greek letters. As usual, the application $t\sigma$ affects only the free
variable occurrences of $t$ whose free variable is found in
$\mathit{dom}_i(\sigma)$ and $\mathit{dom}_f(\sigma)$. Furthermore, we assume
$t\sigma$ to be in $\beta$-normal form unless otherwise stated. A substitution
$\sigma$ is a \emph{unifier} of an SOGU problem $(\mathcal{U},F)$, if
$\mathit{dom}_f(\sigma)= \{F\}$, $\mathit{dom}_i(\sigma)=\emptyset$, and for all
$u\unif_F v\in \mathcal{U}$, $u\sigma=v\sigma$.

The main result presented in the paper concerns \emph{second-order ground
  $E$-unification} where $E$ is a set of equational axioms.
\begin{definition}[Equational theory \cite{DBLP:books/daglib/0092409}]
  \label{def:equtheo}
  Let $E$ be a set of equational axioms. The relation
  \[
    \approx_E  = \{ (s,t) \in \mathcal{T}(\Sigma,\mathcal{V}) \times \mathcal{T}(\Sigma,\mathcal{V}) \; | \; E\models s\approx t\}
  \]
    is called the
  \emph{equational theory} induced by $E$.
\end{definition}
The relation
$\{ (s,t) \in \mathcal{T}(\Sigma,\mathcal{V}) \times \mathcal{T}(\Sigma,\mathcal{V}) \; | \; E\models (s,t)\}$ induced by a set of equalities
$E$ gives the set of equalities satisfied by all structures in the theory of
$E$. We will use the notation $s\approx_E t$ for $(s,t)$ belonging to this set.

In particular, we consider \emph{associative} theories
$$\mathcal{A} = \{f(x,f(y,z))=f(f(x,y),z)\mid f\in \Sigma_{\mathcal{A}}\}$$ where
$\Sigma_{\mathcal{A}}\subset \Sigma$ such that $\Sigma_{\mathcal{A}}$ is finite, possibly empty, and for all
$f\in \Sigma_{\mathcal{A}}$, $\mathit{arity}(f)=2$.  When possible, without causing confusion, we
will use \textit{flattened} notation for function symbols in
$\Sigma_{\mathcal{A}}$, i.e. we write $f(t_1,\ldots,t_n)$ dropping intermediate occurences of $f$.

A \emph{$n$-second-order ground $\mathcal{A}$-unification problem}, abbreviated
$n$-ASOGU problem, is a triple $(\mathcal{U},\mathcal{A},F)$ where
$(\mathcal{U},F)$ is a $n$-second-order ground unification problem and
$\mathcal{A}$ is a non-empty set of associativity axioms. A substitution
$\sigma$ is a \emph{unifier} of an ASOGU problem
$(\mathcal{U},\mathcal{A},F)$, if $\mathit{dom}_f(\sigma)= \{F\}$,
$\mathit{dom}_i(\sigma)=\emptyset$, and for all $u\unif_F v\in \mathcal{U}$, $u\sigma\approx_{\mathcal{A}}v\sigma$.

In later sections, we will use the following theorem due to Matiyasevich,
Robinson, Davis, and Putnam:
\begin{theorem}[ Matiyasevich–Robinson–Davis–Putnam theorem or Hilbert's
  10$^{th}$ problem~\cite{10.5555/164759}]\label{thm:hilbert-10}
  Let $p(\overline{x})$ be a polynomial with integer coefficients. Then whether
  $p(\overline{x}) = 0$ has an integer solution is undecidable.
\end{theorem}
Importantly, \prettyref{thm:hilbert-10} also holds if we restrict the solutions
to non-negative integers (natural numbers with 0).


\section{n-Multipliers and n-Counters}\label{sec:mulcnt}
In this section, we define and discuss the $n$-multiplier and $n$-counter
functions, which allow us to encode number-theoretic problems in second-order
unification. These concepts are related to the \emph{multiplicity operator} used for finiteness results concerning the lambda calculus~\cite{DBLP:journals/jsyml/Vrijer85,DBLP:conf/csl/AccattoliK10}. The results in this section hold for the associative variant of the
unification problem, but we use the syntactic variant below to simplify the
presentation. The $n$-multiplier and $n$-counter are motivated by the following
simple observation about SOGU.
\begin{observation}
  \label{lem:allequal}
  Let $(\mathcal{U},F)$ be a unifiable ASOGU problem, and $\sigma$ a unifier of
  $(\mathcal{U},F)$. Then for all $f\in \Sigma$ and $u\unif_F v\in \mathcal{U}$,
  $\mathit{occ}_{\Sigma}(f,u\sigma) = \mathit{occ}_{\Sigma}(f,v\sigma)$.
\end{observation}
\begin{proof}
Follows from the definition of unifier over an associative theory.
\end{proof}
With this observation, we now seek to relate the number of occurrences of a
symbol $f$ in a term $t$ and a substitution $\sigma$ with the number of occurrences
of $f$ in the term $t\sigma$. The  $n$-\emph{multiplier} counts the multiplicative effect of nested variables, while the  $n$-\emph{counter} counts how the multiplicative effect of nested variables affects the multiplicity of a particular symbol occurring in $t$. 

\begin{definition}[$n$-Mutiplier]
  \label{def:nmulti}
  Let $t$ be an SOG term such that $\Varsf(t)=\{F\}$,
  $F\in \Varsf^n$, and $h_1,\ldots,h_n\geq 0$ are non-negative integers. Then we define the $n$-multiplier for $t$ at $\overline{h_n}$, denoted 
  $\mulOp{F,\overline{h_n},t}$, recursively as follows:
  \begin{itemize}
  \item  $\mulOp{F,\overline{h_n},a} = 0$.\vspace*{.5em}

  \item $\mulOp{F,\overline{h_n}, f(\overline{t_n}))} = \sum_{j=1}^{l}
     \mulOp{F,\overline{h_n},t_j}$\vspace*{.5em}
  \item  $\mulOp{F,\overline{h_n},F(\overline{t_n})} = 1+\sum_{i=1}^{n}
    h_i\cdot\mulOp{F,\overline{h_n},t_i}$\vspace*{.5em}
  \end{itemize}
  \noindent Furthermore, let $(\mathcal{U},F)$ be an SOGU problem. Then \vspace*{.5em}
  \[
    \mullOp{F,\overline{h_n},\mathcal{U}} = \sum_{u\unif_F v\in \mathcal{U}}
    \mulOp{F,\overline{h_n},u} \hspace{2em} \mulrOp{F,\overline{h_n},\mathcal{U}} =
    \sum_{u\unif_F v\in \mathcal{U}} \mulOp{F,\overline{h_n},v}.
  \]
\end{definition}
The \textit{$n$-multiplier} captures the following property of a term: let $t$
be an SOG term such that $\Varsf(t) = \{F\}$, $\mathit{arity}(F)=n$,
$f\in \Sigma$, and $\sigma =\{F\mapsto \lambda \overline{x_n}.s\}$ a substitution where
\begin{itemize}
\item $\mathit{occ}_{\Sigma}(f,s)\geq 1$ and $\mathit{occ}_{\Sigma}(f,t)=0$,
\item $\Varsi(s)\subseteq\{\overline{x_n}\}$, and
\item for all $1\leq i\leq n$, $\mathit{occ}(x_i,s)=h_i$.
\end{itemize}
Then
$\mathit{occ}_{\Sigma}(f,t\sigma)= \mathit{occ}_{\Sigma}(f,s)\cdot
\mulOp{F,\overline{h_n},t}$. Observe that the $\overline{h_n}$ captures
duplication of the arguments to $F$ within the term $t$. Also, it is not a
requirement that $\mathit{occ}_{\Sigma}(f,t)=0$, but making this assumption
simplifies the relationship between $t$ and $t\sigma$ for illustrative
purposes. See the following for a concrete example:

\begin{example}\label{ex:n-mult}
  \begin{figure}
    \centering
    \begin{tikzpicture}[ level 1/.style={level distance = .75cm, sibling
        distance = 3cm}, level 2/.style={level distance = .75cm, sibling
        distance = 1cm}, level 3/.style={level distance = .75cm, sibling
        distance = 1cm}, ]
      \node{g} child {node {\textbf{F}} child {node {g} child {node {a}} child
          {node {\textbf{F}} child {node {s} child {node {a}} } } } } child
      {node {g} child {node {\textbf{F}} child{node {a}} } child {node
          {\textbf{F}} child {node {\textbf{F}} child {node {\textbf{F}}
              child{node {b}} } } } }; 
              
     \draw[dotted] (-3,-.5) rectangle (-.25,-4.5); 
      
      \draw[dotted] (-1.25,-2) rectangle (-.75,-4); 

      \draw[dotted,red, thick]
      (.3,-.5) rectangle (3,-4.5);

      \draw[dotted] (.75,-1.25) rectangle (1.25,-2.5); 
      
      \draw[dotted] (1.75,-1.25) rectangle (2.25,-4);

      \draw[dotted] (-3.5,.25) rectangle (3.75,-5);

      \draw (2,-4.25) node {$1+h_1+h_1^2$}; 
      \draw (1,-2.75) node {1}; 
      \draw (1.5,-4.75) node {$2+h_1+h_1^2$}; 

      \draw (-1,-4.3) node {1};
       
      \draw (-1.65,-4.75) node {$1+h_1$};
      
      \draw (0,-5.25) node
      {$ \mulOp{F,h_1,t} =3+2\cdot h_1+h_1^2$};
    \end{tikzpicture}
    \caption{Computation of $\mulOp{F,h_1,t}$ (See
      Example~\ref{ex:n-mult}).}
    \label{fig:ex1}
  \end{figure}
  Consider the term
  \[
    t=g\bbbbigpl F\bbigpl g\bigpl a,F(s(a))\bigpr\bbigpr,g\bbbigpl F(a),F(F(F(b)))\bbbigpr \bbbbigpr
  \] Then the
  $n$-multiplier of $t$ is $3+2\cdot h_1 + h_1^2$ and is derived as follows:
  \begin{align*}
    \mulOp{F,h_1,t} \ \textbf{=}\ & \mulOp{F,h_1,F\bbigpl g\bigpl a,F(s(a))\bigpr\bbigpr}+ \\ &  \mulOp{F,h_1,{\color{red} g\bbbigpl F(a),F(F(F(b)))\bbbigpr}} \vspace*{.5em} \\
    \mulOp{F,h_1,F\bbigpl g\bigpl a,F(s(a))\bigpr\bbigpr}  \ \textbf{=} \ &  1+ h_1\cdot \mulOp{F,h_1,g(a,F(s(a)))} \vspace*{.5em} \\
    \mulOp{F,h_1,{\color{red} g\bbbigpl F(a),F(F(F(b)))\bbbigpr}} \ \textbf{=} \ & \mulOp{F,h_1,F(a)} + \mulOp{F,h_1,F(F(F(b)))} \vspace*{.5em} \\
    \mulOp{F,h_1,g(a,F(s(a)))}  \ \textbf{=} \ & 1 \vspace*{.5em}\\
    \mulOp{F,h_1,F(a)} \ \textbf{=} \ & 1 \vspace*{.5em}\\
    \mulOp{F,h_1,F(F(F(b)))} \ \textbf{=} \ & 1+h_1\cdot\mulOp{F,h_1,F(F(b))}\vspace*{.5em}\\
    \mulOp{F,h_1,F(F(b)))} \ \textbf{=} \  & 1+h_1 \vspace*{.5em}  
  \end{align*}
  Thus, when $h_1=2$ we get $\mulOp{F,h_1,t} = 11$. Observe
  $\mathit{occ}_{\Sigma}(g',t\{F\mapsto \lambda x.g'(x,x)\}) = 11$, i.e.
  \begin{align*}
    F(g(a,F(s(a))))\{F\mapsto \lambda x.g'(x,x)\} \ \textbf{=} 
    \ & \mathbf{g'}(t',t') \vspace*{.5em} \\
    t' \ \textbf{=} \ & g(a,\mathbf{g'}(s(a),s(a))) \vspace*{.5em} \\
    F(a)\{F\mapsto \lambda x.g'(x,x)\}\ \textbf{=}\ & \mathbf{g'}(a,a)\vspace*{.5em} \\
    F(F(F(b)))\{F\mapsto \lambda x.g'(x,x)\} \ \textbf{=} \ & \mathbf{g'}(t'',t'')\vspace*{.5em} \\
    t''\ \textbf{=}\ & \mathbf{g'}(\mathbf{g'}(b,b),\mathbf{g'}(b,b))\vspace*{.5em}
  \end{align*}
  Given that $\mathit{occ}_{\Sigma}(g',t) = 0$, all occurrences of $g'$ are
  introduced by $\sigma$. See Figure~\ref{fig:ex1} for a tree representation of
  the computation.

\end{example}

Next, we introduce the $n$-counter function. Informally, given an SOG term
$t$ such that $\Varsf(t)=\{F\}$, a symbol $g\in \Sigma$, and a
substitution $\sigma$ with $\mathit{dom}_f(\sigma)=\{F\}$, the $n$-counter
captures the number of occurrences of $g$ in $t\sigma$. The substitution can change the occurrences of 
$g$ in two ways: (i) by introducing additional occurrences of $g$ and (ii) by duplicating occurrences of $g$ that occur as arguments to $F$. The  $n$-counter function captures both of these changes. In particular, we are interested in occurrences of the latter type as these are exploited by the reduction presented in Section~\ref{app:undecidable}. 
\begin{definition}[$n$-Counter]
  \label{def:ncounter}
  Let $g\in \Sigma$, $t$ be an SOG term such that $\Varsf(t)=\{F\}$
  and $F\in \Varsf^n$, and $h_1,\ldots,h_n\geq 0$ are non-negative integers. Then we define the $n$-counter of $g$ for 
  $t$ at $\overline{h_n}$, denoted
  $\cntOp{F,\overline{h_n},g,t}$, recursively as follows:
  \begin{itemize}
  \item $\cntOp{F,\overline{h_n},a,b} = 0$ \vspace*{.5em}
  \item $\cntOp{F,\overline{h_n},a,a} = 1$\vspace*{.5em}
  \item $\cntOp{F,\overline{h_n},f, g(\overline{t_l})} = \sum_{j=1}^{l}
    \cntOp{F,\overline{h_n},f,t_j}$\vspace*{.5em}
  \item $\cntOp{F,\overline{h_n},f,f(\overline{t_l})} = 1+ \sum_{j=1}^{l}
    \cntOp{F,\overline{h_n},f,t_j}$\vspace*{.5em}
  \item $\cntOp{F,\overline{h_n},f,F(\overline{t_n})} = \sum_{i=1}^{n}
    h_i\cdot\cntOp{F,\overline{h_n},f,t_i}$\vspace*{.5em}
  \end{itemize}
  \noindent Furthermore, let $(\mathcal{U},F)$ be an SOGU problem. Then\vspace*{.5em}
  \[
    \cntlOp{F,\overline{h_n},g,\mathcal{U}} = \sum_{u\unif_F v\in
      \mathcal{U}} \cntOp{F,\overline{h_n},g,u} \hspace{3em}
    \cntrOp{F,\overline{h_n},g,\mathcal{U}} = \sum_{u\unif_F v\in
      \mathcal{U}} \cntOp{F,\overline{h_n},g,v}.
  \]
\end{definition}

The \textit{$n$-counter} captures the following property of a term: let $t$ be an SOG term such that $\Varsf(t) = \{F\}$, $\mathit{arity}(F)=n$,
$f\in \Sigma$, and $\sigma =\{F\mapsto \lambda \overline{x_n}.s\}$ a
substitution where
\begin{itemize}
\item $\mathit{occ}_{\Sigma}(f,s)= 0$,
\item $\Varsi(s)\subseteq\{\overline{x_n}\}$, and
\item for all $1\leq i\leq n$, $\mathit{occ}_{\Sigma}(x_i,s)=h_i$.
\end{itemize}
Then $\mathit{occ}_{\Sigma}(f,t\sigma)=
\cntOp{F,\overline{h_n},f,t}$. Observe that the $\overline{h_n}$ captures
duplication of the arguments to $F$ within the term $t$. Also, it is not a
requirement that $\mathit{occ}_{\Sigma}(f,s)=0$, but making this assumption
simplifies the relationship between $t$ and $t\sigma$ for illustrative
purposes. See the following for a concrete example:

\begin{example}\label{ex:n-counter}
  \begin{figure}
    \centering
    \begin{tikzpicture}[ level 1/.style={level distance = .75cm, sibling
        distance = 3cm}, level 2/.style={level distance = .75cm, sibling
        distance = 3cm}, level 3/.style={level distance = .75cm, sibling
        distance = 1cm}, ]
      \node{g} child { node{a} } child {
        node{g} child {node {\textbf{F}} child {node {g} child {node {a}} child
            {node {\textbf{F}} child {node {g} child {node {a}} child {node {a}}
              } } } } child {node {g} child {node {\textbf{F}} child{node {a}} }
          child {node {\textbf{F}} child {node {\textbf{F}} child {node
                {\textbf{F}} child{node {b}} } } } } } ; 
                
      \draw[dotted] (-1.75,-.5) rectangle (-1.25,-1); 
      \draw (-1.5,-1.25) node {$1$};

      \draw[dotted] (-.25,-2.75) rectangle (1.25,-4.75); 
      \draw (.5,-5) node      {$2\cdot h_1$};

      \draw[dotted] (2.25,-2) rectangle (2.75,-3.25); 
      \draw (2.5,-3.5) node      {$h_1$};

      \draw[dotted] (3.25,-2) rectangle (3.75,-4.75); \draw (3.5,-5) node {$0$};

      \draw[dotted] (-1,-.5) rectangle (4.65,-5.75); 
      \draw (2,-6) node {$ 2\cdot h_1+2\cdot h_1^2$};

      \draw[dotted] (-.75,-1.25) rectangle (1.5,-5.25); \draw (.5,-5.5) node
      {$h_1+2\cdot h_1^2$};

      \draw[dotted,red,thick] (1.9,-1.25) rectangle (4.15,-5.25); 
      \draw (3,-5.5) node {$h_1$};

      \draw[dotted] (-2,.25) rectangle (5,-6.25); 
      
      \draw (1.5,-6.5) node {$ \cntOp{F,h_1,a,t} =1+2\cdot h_1+2\cdot h_1^2$};
    \end{tikzpicture}
    \caption{Computation of $\cntOp{F,h_1,a,t}$ (See
      Example~\ref{ex:n-counter}).}
    \label{fig:ex2}
  \end{figure}

  Consider the term $t=$
  \[
    g\bbbbigpl a,g\bbbigpl F\bbigpl g\bigpl a,F(g(a,a))\bigpr \bbigpr,g\bbigpl F(a),F(F(F(b)))\bbigpr\bbbigpr\bbbbigpr.
  \]
  The counter of $t$ over the symbol $a$ is $1+ 2\cdot h_1 + 2\cdot h_1^2$ and is derived as follows:
  \begin{align*}  
  \cntOp{F,h_1,a,t}\ \textbf{=}\ &
  \cntOp{F,h_1,a,a}+ \\ &
 \cntOp{F,h_1,a,F\bbigpl g\bigpl a,F(g(a,a))\bigpr \bbigpr} +\\
 &\cntOp{F,h_1,a,g\bbigpl F(a),F(F(F(b)))\bbigpr} \vspace*{.5em} \\
  \cntOp{F,h_1,a,a}\ \textbf{=}\ & 1\vspace*{.5em} \\
  \cntOp{F,h_1,a,F\bbigpl g\bigpl a,F(g(a,a))\bigpr \bbigpr} \ \textbf{=}\ &
 h_1\cdot \cntOp{F,h_1,a,{\color{red} g\bbigpl F(a),F(F(F(b)))\bbigpr}}\vspace*{.5em} \\
   \cntOp{F,h_1,a,{\color{red}g(F(a),F(F(F(b))))}}  \ \textbf{=}\ &
   \cntOp{F,h_1,a,F(a)} +  \cntOp{F,h_1,a,F(F(F(b)))} \vspace*{.5em} \\
  \cntOp{F,h_1,a,g(a,F(g(a,a)))}\ \textbf{=}\ & 1+ 2\cdot h_1\vspace*{.5em} \\
  \cntOp{F,h_1,a,F(a)}\ \textbf{=}\ & h_1 \vspace*{.5em}\\
  \cntOp{F,h_1,a,F(F(F(b)))}\ \textbf{=}\ & 0 \vspace*{.5em}
  \end{align*}

  \noindent Thus, when $h_1=2$ we get $\cntOp{F,h_1,a,t} = 13$ (the sum of the three main components as illustrated above). Observe
  $\mathit{occ}_{\Sigma}(a,t\{F\mapsto \lambda x.g(x,x)\}) = 13$. Applying the
  substitution $\sigma =\{F\mapsto \lambda x.g(x,x)\}$ to $t$ results in the
  following
  \begin{align*}
    F(g(a,F(g(a,a))))\sigma \ \textbf{=}\ & g(t',t')\vspace*{.5em} \\
    t' \ \textbf{=}\ & g(\textbf{a},g(g(\textbf{a},\textbf{a}),g(\textbf{a},\textbf{a})))\vspace*{.5em} \\
    F(a)\sigma\ \textbf{=}\ & g(\textbf{a},\textbf{a})\vspace*{.5em} \\
    F(F(F(b)))\sigma \ \textbf{=}\ & g(t'',t'')\vspace*{.5em} \\
    t'' \ \textbf{=}\ & g(g(b,b),g(b,b))\vspace*{.5em}
  \end{align*}
  See Figure~\ref{fig:ex2} for a tree representation of the computation.
\end{example}

The $n$-multiplier and $n$-counter operators count the occurrences of symbols in
$t\sigma$ by only considering $t$. Observe that $h_1,\ldots, h_n$ denote the multiplicity
of the arguments of the function variable $F$ within a substitution $\sigma$ with
domain $F$. For any symbol $c$ occurring in $t$, the $n$-counter predicts the
occurences of $c$ in $t\sigma$ by only considering $t$ and $h_1,\ldots, h_n$.  For any
symbol $c$ occurring in the range of a substitution $\sigma$, the $n$-multiplier
predicts the occurences of $c$ in $t\sigma$ by only considering $t$ and
$h_1,\ldots, h_n$. We now describe the relationship between the $n$-multiplier,
$n$-counter, and the total occurrences of a given symbol within the term $t\sigma$.

\begin{lemma}
  \label{lem:FixingConstants}
  Let $g\in \Sigma$, $n\geq 0$, $t$ be an SOG term such that
  $\Varsf(t)=\{F\}$, $h_1,\ldots,h_n\geq 0$ are non-negative integers, and
  $\sigma = \{F\mapsto \lambda\overline{x_n}.s\}$ a substitution such that
  $\Varsi(s)\subseteq\{\overline{x_n}\}$ and for all $1\leq i\leq n$
  $\mathit{occ}_{\Sigma}(x_i,s) = h_i$. Then
  $$\mathit{occ}_{\Sigma}(g,t\sigma) =
  \mathit{occ}_{\Sigma}(g,s)\cdot\mulOp{F,\overline{h_n},t}+\cntOp{F,\overline{h_n},g,t}.$$
\end{lemma}
\begin{proof}
  We prove the lemma by induction on $\dep{t}$.  \proofcase{Base case:} When
  $\dep{t}=1$, then either (i) $t$ is a constant or (ii) $t=F$ and
  $\mathit{arity}(F)=0$.
    \begin{enumerate}[label=(\roman*)]
  \item Observe that $t=t\sigma$ and $\mulOp{F,\overline{h_n},t}= 0$. If
    $t=g$ then $\cntOp{F,\overline{h_n},g,t} = 1$, otherwise
    $\cntOp{F,\overline{h_n},g,t} = 0$. In either case, we get
    $\mathit{occ}(g,t\sigma)_{\Sigma} = 0+\cntOp{F,\overline{h_n},g,t}$.
  \item Observe that $t\sigma = s$, $\mulOp{F,t}= 1$, and
    $\cntOp{F,g,t} = 0$. Thus,
    $\mathit{occ}(g,t\sigma)_{\Sigma} = \mathit{occ}_{\Sigma}(g,s)\cdot\mulOp{F,t}$ which reduces
    to $\mathit{occ}_{\Sigma}(g,s)=\mathit{occ}_{\Sigma}(g,s) $.
  \end{enumerate}
  Hence, in either case, we obtain the desired result.%
  \proofcase{Step case:} Now, for the induction hypothesis, we assume the lemma
  holds for all terms $t$ such that $\dep{t}<w+1$ and prove the lemma for a term
  $t'$ such that $\dep{t}=w+1$. Consider the following three cases:
  \begin{itemize}
  \item $t=f(t_1,\ldots,t_k)$ and $f = g$. We know by the induction hypothesis that
    for $1\leq i\leq k$,
    $\mathit{occ}_{\Sigma}(g,t_i\sigma) =
    \mathit{occ}_{\Sigma}(g,s)\cdot\mulOp{F,\overline{h_n},t_i}+\cntOp{F,\overline{h_n},g,t_i}.$
    Thus,
    \begin{align*}
      \mathit{occ}_{\Sigma}(g,t\sigma)\ \textbf{=}   1+\sum_{i=1}^{k}
      \mathit{occ}_{\Sigma}(g,t_i\sigma)\ \textbf{=} &\\ 1+ \mathit{occ}_{\Sigma}(g,s)\cdot\sum_{i=1}^{k} \mulOp{F,\overline{h_n},t_i} + \sum_{i=1}^k \cntOp{F,\overline{h_n},g,t_i}\ \textbf{=} &\\
      \mathit{occ}_{\Sigma}(g,s)\cdot \mulOp{F,\overline{h_n},t} +  \cntOp{F,\overline{h_n},g,t}
    \end{align*}
    where, by the definition of the $n$-multiplier and the $n$-counter,
    $\mulOp{F,\overline{h_n},t} =$\\
    $ \sum_{i=1}^{k} \mulOp{F,\overline{h_n},t_i},$ and
    $\cntOp{F,\overline{h_n},g,t} = 1+\sum_{i=1}^k \cntOp{F,\overline{h_n},g,t_i}.$

  \item $t=f(t_1,\ldots,t_k)$ and $f\not = g$: This case follows from the previous
    case except, we do not count $f$ when counting occurrences of $g$ and
    $\cntOp{F,\overline{h_n},g,t} = \sum_{i=1}^k \cntOp{F,\overline{h_n},g,t_i}.$
  \item $t=F(r_1,\ldots,r_n)$. By the induction hypothesis, we have that for all
    $1\leq i\leq n$,
    \begin{align*}
      &\mathit{occ}_{\Sigma}(g,r_i\sigma)\ \textbf{=}\  \mathit{occ}_{\Sigma}(g,s)\cdot\mulOp{F,\overline{h_n},r_i}+\cntOp{F,\overline{h_n},g,r_i}.
    \end{align*}
    We can derive the following equality and conclude the proof using the above
    assumption.
    \begin{align*}
      \mathit{occ}_{\Sigma}(g,t\sigma)\ \textbf{=}\ \mathit{occ}_{\Sigma}(g,s)+\sum_{i=1}^{n} h_i\cdot\mathit{occ}_{\Sigma}(g,r_i\sigma)\ \textbf{=} \  &  \\ \vspace{1em}
      \mathit{occ}_{\Sigma}(g,s)+\mathit{occ}_{\Sigma}(g,s)\cdot \left(\sum_{i=1}^{n} h_i\cdot \mulOp{F,\overline{h_n},r_i}\right)+ 
      \left(\sum_{i=1}^{n} h_i\cdot\cntOp{F,\overline{h_n},g,r_i}\right)\ \textbf{=}  & \\ \vspace{1em}
      \mathit{occ}_{\Sigma}(g,s)\cdot \left(1+ \sum_{i=1}^{n} h_i\cdot \mulOp{F,\overline{h_n},r_i}\right)+ \left(\sum_{i=1}^{n} h_i\cdot\cntOp{F,\overline{h_n},g,r_i}\right)\ \textbf{=} & \\ \vspace{1em} 
        \mathit{occ}_{\Sigma}(g,s) \cdot \mulOp{F,\overline{h_n},F(\overline{r_n})}+  \cntOp{F,\overline{h_n},g,F(\overline{r_n})}\ \textbf{=} & \\  \vspace*{1em} \mathit{occ}_{\Sigma}(g,s) \cdot \mulOp{F,\overline{h_n},t}+ \cntOp{F,\overline{h_n},g,t} \hspace*{1.2em} & \vspace{.5em} \tag*{\qedhere}
    \end{align*}
  \end{itemize}
\end{proof}
When we use Lemma~\ref{lem:FixingConstants} in Section~\ref{app:undecidable}, the term $s$ in $\sigma = \{F\mapsto \lambda\overline{x_n}.s\}$ we will not contain any occurences of $g$, i.e. $\mathit{occ}_{\Sigma}(g,s)=0$. Furthermore, Lemma~\ref{lem:FixingConstants} captures an essential property of the $n$-multiplier and $n$-counter,
which we illustrate in the following example.
\begin{example}
  Consider the term $t=$
\[
    g\bbbbigpl a,g\bbbigpl F\bbigpl g\bigpl a,F(g(a,a))\bigpr \bbigpr,g\bbigpl F(a),F(F(F(b)))\bbigpr\bbbigpr\bbbbigpr.
  \]
  and the substitution $\{F\mapsto \lambda x.g(a,g(x,x))\}$. The $n$-counter of $a$ for 
  $t$ at 2 is $\cntOp{F,2,a,t} = 13$ and the $n$-multiplier for $t$ at 2 is
  $\mulOp{F,2,t} = 11$. Observe
  $\mathit{occ}_{\Sigma}(a,t\{F\mapsto \lambda x.g(a,\allowbreak g(x,x))\}) = 24$ and
  $\mathit{occ}_{\Sigma}(a,s)\cdot\mulOp{F,2,t}+\cntOp{F,2,a,t} = 24$. Applying the
  substitution $\{F\mapsto \lambda x.g(a,g(x,x))\}$ to $t$ results in the following:
  \begin{align*}
    F(g(a,F(g(a,a))))\{F\mapsto \lambda x.g(a,g(x,x))\}\ \textbf{=}\ & g(\textbf{a},g(t',t'))\\
    t'\ \textbf{=}\ & g(\textbf{a},g(\textbf{a},g(g(\textbf{a},\textbf{a}),g(\textbf{a},\textbf{a}))))\\
    F(a)\{F\mapsto \lambda x.g(a,g(x,x))\} \ \textbf{=}\ & g(\textbf{a},g(\textbf{a},\textbf{a}))\\
    F(F(F(b)))\{F\mapsto \lambda x.g(a,g(x,x))\} \ \textbf{=}\ & g(\textbf{a},g(t'',t''))\\
    t'' \ \textbf{=}\ & g(\textbf{a},g(g(\textbf{a},g(b,b)),g(\textbf{a},g(b,b))))
  \end{align*}
\end{example}
So far, we have only considered arbitrary terms and substitutions. We now apply the above results to unification problems and their solutions. In particular, a
corollary of Lemma~\ref{lem:FixingConstants} is that there is a direct relation
between the $n$-multiplier and $n$-counter of an unifiable unification
problem. The following lemma describes this relation.
\begin{lemma}[Unification Condition]\label{lem:MainEq}
  Let $(\mathcal{U},F)$ be an unifiable SOGU problem such that
  $\Varsf(\mathcal{U}) = \{F\}$, $h_1,\ldots,h_n\geq 0$ are non-negative integers, and
  $\sigma =\{F\mapsto \lambda \overline{x_n}.s\}$ an unifier of
  $(\mathcal{U},F)$ such that $\Varsi(s)\subseteq \{\overline{x_n}\}$ and for all
  $1\leq i\leq n$,$\mathit{occ}_{\Sigma}(x_i,s)=h_i$. Then for all $g\in \Sigma$,\vspace{.5em}
  \begin{equation}\label{eq:perunifeqCom}
    \begin{array}{l}
      \mathit{occ}_{\Sigma}(g,s)\cdot (\mullOp{F,\overline{h_n},\mathcal{U}}-\mulrOp{F,\overline{h_n},\mathcal{U}})\ \textbf{=} \ \cntrOp{F,\overline{h_n},g,\mathcal{U}} -\cntlOp{F,\overline{h_n},g,\mathcal{U}}.\end{array}
  \end{equation}
\end{lemma}
\begin{proof}
  By Lemma~\ref{lem:allequal}, for any $g\in \Sigma$ and
  $u\unif_F v\in \mathcal{U}$, we have
  $\mathit{occ}_{\Sigma}(g,u\sigma) = \mathit{occ}_{\Sigma}(g,v\sigma)$ and by
  Lemma~\ref{lem:FixingConstants} we also have
  \begin{gather*}
    \mathit{occ}_{\Sigma}(g,u\sigma) =
    \mathit{occ}_{\Sigma}(g,s)\cdot\mulOp{F,\overline{h_n},u}+\cntOp{F,\overline{h_n},g,u},\\
    \text{ and}\\
    \mathit{occ}_{\Sigma}(g,v\sigma) = \mathit{occ}_{\Sigma}(g,s)\cdot\mulOp{F,\overline{h_n},v}+\cntOp{F,\overline{h_n},g,v}.
  \end{gather*}
  Thus, for any $g\in \Sigma$ and $u\unif_F v\in \mathcal{U}$,
  \begin{align*}
    \mathit{occ}_{\Sigma}(g,s)\cdot\mulOp{F,\overline{h_n},u}+\cntOp{F,\overline{h_n},g,u}    \ \textbf{=}\ 
    \mathit{occ}_{\Sigma}(g,s)\cdot\mulOp{F,\overline{h_n},v}+\cntOp{F,\overline{h_n},g,v}.
  \end{align*}
  From this equation, we can derive the following:
  \begin{equation} \label{eq:unif2}
    \begin{gathered}
      \mathit{occ}_{\Sigma}(g,s)\cdot(\mulOp{F,\overline{h_n},u}-\mulOp{F,\overline{h_n},v})
      \ \textbf{=}\ 
      \cntOp{F,\overline{h_n},g,v}-\cntOp{F,\overline{h_n},g,u}
    \end{gathered}
  \end{equation}
  We can generalize Equation~\ref{eq:unif2} to $\mathcal{U}$ by computing
  Equation~\ref{eq:perunifeqCom} for each $u\unif_F v\in \mathcal{U}$ and adding the results
  together and thus deriving the following
  \begin{equation*}
    \begin{gathered}
      \mathit{occ}_{\Sigma}(g,s)\cdot (\mullOp{F,\overline{h_n},\mathcal{U}}-\mulrOp{F,\overline{h_n},\mathcal{U}})    \ \textbf{=}\ 
        \cntrOp{F,\overline{h_n},g,\mathcal{U}}-\cntlOp{F,\overline{h_n},g,\mathcal{U}}. \tag*{\qedhere}
    \end{gathered}
  \end{equation*}
\end{proof}

The \emph{unification condition} provides a necessary condition for unifiability
that we use in the undecidability proof presented in
Section~\ref{app:undecidable}, specifically the relationship between the left
and right sides of the unification equation presented in
Equation~\ref{eq:perunifeqCom}. Sufficiency requires an additional assumption,
namely, the signature contains at least one \emph{associative function
  symbol}. The following example shows an instance of this property.
\begin{example}
  Consider the SOGU problem $F(g(a,a))\unif_F g(F(a),F(a))$ and the unifier
  $\sigma=\{F\mapsto \lambda x. g(x,x)\}$. Observe that
  \begin{align*}
    \mathit{occ}_{\Sigma}(a,g(x,x)) \cdot (\ \mullOp{F,2,F(g(a,a))}- \mulrOp{F,2,g(F(a),F(a))}\ )
    \ \textbf{=}\  0\cdot(1-2)  \ \textbf{=}\  0
  \end{align*}
  and for the right side of Equation~\ref{eq:perunifeqCom} we get
  \begin{gather*}
    \cntrOp{F,h,a,g(F(a),F(a))} -\cntlOp{F,h,a,F(g(a,a))} \ \textbf{=}\  4-4  \ \textbf{=}\  0.
  \end{gather*}
\end{example}

\section{Undecidability of ASOGU Unification}
\label{app:undecidable}
\begin{figure*}
  \centering
  \begin{tikzpicture}[ level 1/.style={level distance = 1.1cm, sibling distance
      = 4cm}, level 2/.style={level distance = 1.1cm, sibling distance = 4cm},
    level 3/.style={level distance = 1.1cm, sibling distance = 3.5cm}, level
    4/.style={level distance = 1.1cm, sibling distance = 2cm}, ]
    \node{${\color{red} x^2yz - 7xyz - x^2z + 7xz - 14x - 2x^2y + 2x^2 + 14xy}$}
    child {node
      {${\color{red}xyz} {\color{blue}- 7yz} {\color{red}- xz} {\color{green} +
          7z} - 14 {\color{red}- 2xy} {\color{red}+ 2x} {\color{blue}+ 14y}$}
      child {node {$-14$} edge from parent node[circle,fill=white] {{\small 0}}
      } child {node {{\color{red}$xyz - xz - 2xy + 2x$}} child { node
          {$ {\color{blue}yz} {\color{green}- z} {\color{blue}- 2y} + 2$} child
          {node {$2$} edge from parent node[circle,fill=white] {{\small 0}} }
          child {node {{\color{blue}$yz- 2y$}} child {node
              {{\color{green}$z-2$}} child {node {$-2$} edge from parent
                node[circle,fill=white] {{\small 0}} } child {node
                {{\color{green}$z$}} child {node {$1$} } edge from parent
                node[circle,fill=white] {z} } } edge from parent
            node[circle,fill=white,inner sep=1pt] {y} } child {node
            {{\color{green}$- z$}} child {node {$-1$}} edge from parent
            node[circle,fill=white] {z} } } edge from parent
        node[circle,fill=white] {x}} child {node {{\color{blue}$7yz + 14y$}}
        child {node {$ {\color{green}7z} + 14$} child {node {$14$} edge from
            parent node[circle,fill=white] {{\small 0}} } child {node
            {{\color{green}$7z$}} child {node {$7$}} edge from parent
            node[circle,fill=white] {z} } } edge from parent
        node[circle,fill=white] {y} } child {node {{\color{green}$7z$}} child
        {node {7}} edge from parent node[circle,fill=white] {z} } }; 
        
    \draw (-4.5,.25) rectangle (4.5,-1.25);
    \draw (-3.8,-2) rectangle (-.2,-3.5); 
    \draw (1,-2) rectangle (3,-3.5); 
    \draw (5.7,-2) rectangle (6.3,-3.5); 
    \draw (-2.8,-4.2) rectangle (-1.2,-5.7);
    \draw (-.3,-4.1) rectangle (.3,-5.7); 
    \draw (2.7,-4.2) rectangle (3.3,-5.7);
    \draw (-1.3,-6.4) rectangle (-.7,-7.9);
  \end{tikzpicture}
  \caption{We recursively apply reduction and monomial grouping decomposition
    (Definition~\ref{def:monogrp}) to the polynomial at the root of the tree. In each box, the lower polynomial is the reduction of the upper polynomial by the unknown labeling the edge to the parent box. By {\small 0}, we denote the monomial grouping 0, and x,y, and z denote the groupings associated with unknowns.}
  \label{fig:recdecom}
\end{figure*}

We now use the machinery we built in the previous section to encode
\emph{Diophantine equations} in unification problems. As a result, we can
transfer undecidability results for Diophantine equations to ASOGU
unification problems. Our undecidability result hinges on
Lemma~\ref{lem:allequal}. Observe that for SOGU, two terms might have an
equal number of occurrences of a symbol without being syntactically
equal. Hence, we introduce an associative binary
function symbol $g$ to solve this issue. For the remainder of this section, we
consider a finite signature $\Sigma$ such that $\{g,a\}\subseteq \Sigma$,
$\mathit{arity}(g) = 2$, $\mathit{arity}(a) = 0$, and
$\mathcal{A}=\{g(x,g(y,z))= g(g(x,y),z)\}$. We will write $g$ in \textit{flattened form}
(see Section~\ref{sect:prelims}). Note that since our signature only consists of
a single constant, strictly speaking, we only require the \emph{weaker} property
of \emph{power associativity}.

We now introduce the basic definitions needed to describe our translation from
polynomials to terms. By $p(\overline{x_n})$ we denote a polynomial in reduced
form\footnote{multiplication is fully distributed over addition and combining
  like terms.} with integer coefficients over the unknowns $x_1,\ldots,x_n$ ranging
over the natural numbers and by $\mathit{mono}(p(\overline{x_n}))$ we denote the
set of monomials of $p(\overline{x_n})$. Given a polynomial $p(\overline{x_n})$
and $1\leq i \leq n$ by $\mathit{div}(p(\overline{x_n}),x_i)$ we denote that
$x_i$ \textit{divides} $p(\overline{x_n})$.
Furthermore,
$\mathit{deg}(p(\overline{x}_n)) = \max\{ k\mid k\geq 0\wedge m=x_i^{k}\cdot q(\overline{x_n})
\wedge 1\leq i \leq n\wedge m\in \mathit{mono}(p(\overline{x_n}))\}$.  Given a polynomial
$p(\overline{x_n})$, a polynomial $p'(\overline{x_n})$ is a sub-polynomial of
$p(\overline{x_n})$ if
$\mathit{mono}(p'(\overline{x_n}))\subseteq \mathit{mono}(p(\overline{x_n}))$. Using the above definition, we define distinct sub-polynomials based on divisibility by
one of the input unknowns (Definition~\ref{def:monogrp}). See
Figure~\ref{fig:recdecom} for an illustration of the procedure defined in
Definition~\ref{def:monogrp} recursively applied to a polynomial.
\begin{definition}[monomial groupings]
  \label{def:monogrp}
  Let $p(\overline{x_n})= q(\overline{x_n})+c$ be a polynomial where
  $q(\overline{x_n})$ does not have a constant term and $c\in \mathbb{Z}$,
  $<_{p(\overline{x_n})}$ be a total linear ordering on $\overline{x_n}$, and
  for all $0\leq j\leq n$,
  $S_{x_j} = \{m \mid m\in \mathit{mono}(p(\overline{x_n})) \wedge \forall i
  (1\leq i\leq n\wedge x_i<_{p(\overline{x_n})} x_j \Rightarrow \neg
  \mathit{div}(m,x_i))\}$, that is the set of monomials not divisible by
any variable smaller than $x_j$. Then

  \begin{itemize}
  \item $p(\overline{x_n})_0 = c$,
  \item $p(\overline{x_n})_{x_j} = 0$ if there does not exist $m\in S_{x_j}$
    such that $\mathit{div}(m,x_j)$,
  \item otherwise, $p(\overline{x_n})_{x_j} = p'(\overline{x_n})$, where
    $p'(\overline{x_n})$ is the sub-polynomial of $p(\overline{x_n})$ such that
    $\mathit{mono}(p'(\overline{x_n}))= \{ m\mid m\in S_{x_j} \wedge
    \mathit{div}(m,x_j)\}$.
  \end{itemize}
  Furthermore, let $p(\overline{x_n})_{x_j} = x_j\cdot p'(\overline{x_n})$. Then
  $p(\overline{x_n})_{x_j} \downarrow = p'(\overline{x_n})$ is the
  \emph{reduction} of $p(\overline{x_n})_{x_j}$, i.e., the degree of $x_j$ in $p(\overline{x_n})_{x_j}$ is reduced by one.
\end{definition}
Essentially, monomial groupings are a way to partition a given polynomial
$p(\overline{x_n})$ with respect to an ordering $<_{p(\overline{x_n})}$ on the
unknowns $\overline{x_n}$. This partition results in a set of subpolynomials of
$p(\overline{x_n})$ and for each subpolynomial, its monomials share a common
unknown, thus implying that a reduction of these subpolynomials always exists.
\begin{example}
  Consider the polynomial
  \begin{align*}
    p(x,y,z) &= x^2yz - 7xyz - x^2z + 7xz - 12z\\
             &- 14x - 2x^2y + 2x^2 + 14xy + 12yz - 24y + 24
  \end{align*}
  Assuming the unknowns are ordered $x<_{p(x,y,z)}y<_{p(x,y,z)}z$, its monomial
  grouping are as follows:
  \begin{align*}
    p(x,y,z)_0 =&\ 24\\
    p(x,y,z)_x =&\ 
                  x^2yz - 7xyz - x^2z + 7xz  - 14x - 2x^2y + 2x^2 + 14xy\\
    p(x,y,z)_y =&\  12yz - 24y\\
    p(x,y,z)_z =&\  - 12z
  \end{align*}
  The reduction of these groupings is as follows:
  \begin{align*}
    p(x,y,z)_x\downarrow =&\  xyz - 7yz - xz + 7z  - 14 - 2xy + 2x + 14y\\
    p(x,y,z)_y\downarrow =&\  12z - 24\\
    p(x,y,z)_z\downarrow =&\  - 12
  \end{align*}
  In Figure~\ref{fig:recdecom}, we recursively apply reduction and grouping
  decomposition to $p(x,y,z)_x$. The illustrated procedure is at the heart of
  our encoding (Definition~\ref{def:convert}).
\end{example}
We now define a second-order term representation for arbitrary polynomials as
follows:
\begin{definition}[$n$-Converter]
  \label{def:convert}
  Let $p(\overline{x_n})$ be a polynomial and $F\in \Varsf^n$. The
  positive and negative second-order term representation of $p(\overline{x_n})$
  denoted as $\cvtpOp{F,p(\overline{x_n})}$ and $\cvtnOp{F,p(\overline{x_n})}$
  respectively are recursively defined as follows:
  \begin{itemize}
  \item if $p(\overline{x_n}) = p(\overline{x_n})_0=0$, then
    $$\cvtpOp{F,p(\overline{x_n})}=\cvtnOp{F,p(\overline{x_n})} =
    a$$ \vspace{.5em}
  \item if $p(\overline{x_n}) = p(\overline{x_n})_0=c\geq 1$, then \vspace{.5em}
    \begin{itemize}
    \item
      $\cvtpOp{F,p(\overline{x_n})}=g(\overbrace{a,\ldots,a}^{\mathclap{|p(\overline{x_n})_0|+1
        \mbox{ occurences}}})$ \vspace{.5em}

    \item $\cvtnOp{F,p(\overline{x_n})}=a$.  \vspace{.5em}
    \end{itemize}
  \item if $p(\overline{x_n}) = p(\overline{x_n})_0<0$, then \vspace{.5em}
    \begin{itemize}
    \item $\cvtpOp{F,p(\overline{x_n})}=a$.\vspace{.5em}
    \item
      $\cvtnOp{F,p(\overline{x_n})}=g(\overbrace{a,\ldots,a}^{\mathclap{|p(\overline{x_n})_0|+1
        \mbox{ occurences}}})$ \vspace{.5em}
    \end{itemize}

  \item if $p(\overline{x_n}) \not = p(\overline{x_n})_0$ and
    $p(\overline{x_n})_0=0$, then for $\star \in \{+,-\}$,
    $\cvtsOp{F,p(\overline{x_n})}= F(t_1,\ldots,t_n)$\vspace*{.5em}
    \begin{itemize}
    \item where for all $1\leq i\leq n$,
      $t_i =\cvtsOp{F,p(\overline{x_n})_{x_i}\downarrow}$.\vspace*{.5em}
    \end{itemize}

  \item if $p(\overline{x_n}) \not = p(\overline{x_n})_0$ and
    $p(\overline{x_n})_0\geq 1$, then \vspace{.5em}
    \begin{itemize}
    \item $\cvtpOp{F,p(\overline{x_n})}=g(t,F(t_1,\ldots,t_n))$ where \vspace{.5em}
      \begin{itemize}
      \item
        $t=\overbrace{a,\ldots,a}^{\mathclap{|p(\overline{x_n})_0| \mbox{ occurences}}}$ and~\footnote{Remember that nested associative functions are written in flattened form.}
        \vspace{.5em}
      \item for all $1\leq i\leq n$,
        $t_i = \cvtpOp{F,p(\overline{x_n})_{x_i}\downarrow}$.  \vspace{.5em}
      \end{itemize}
    \item $\cvtnOp{F,p(\overline{x_n})}=F(t_1,\ldots,t_n)$ where \vspace{.5em}
      \begin{itemize}

      \item for all $1\leq i\leq n$,
        $t_i= \cvtnOp{F,p(\overline{x_n})_{x_i}\downarrow}$.\vspace{.5em}
      \end{itemize}

    \end{itemize}
  \item if $p(\overline{x_n}) \not = p(\overline{x_n})_0$ and
    $p(\overline{x_n})_0< 0$, then
    \begin{itemize}
    \item $\cvtpOp{F,p(\overline{x_n})}=F(t_1,\ldots,t_n)$ where \vspace{.5em}
      \begin{itemize}
      \item for all $1\leq i\leq n$,
        $t_i=\cvtpOp{F,p(\overline{x_n})_{x_i}\downarrow}$.\vspace{.5em}
      \end{itemize}
    \item $\cvtnOp{F,p(\overline{x_n})}=g(t,F(t_1,\ldots,t_n))$ where \vspace{.5em}
      \begin{itemize}
      \item
        $t=\overbrace{a,\ldots,a}^{\mathclap{|p(\overline{x_n})_0| \mbox{ occurences}}}$ and
        \vspace{.5em}
      \item for all $1\leq i\leq n$,
        $t_i = \cvtnOp{F,p(\overline{x_n})_{x_i}\downarrow}$.  \vspace{.5em}
      \end{itemize}
    \end{itemize}

  \end{itemize}
\end{definition}
Intuitively, the $n$-converter takes a polynomial in $n$ unknowns and separates
it into $n+1$ subpolynomials disjoint with respect to monomial groupings. Each
of these sub-polynomials is assigned to one of the arguments of the second-order
variable (except the subpolynomial representing an integer constant), and the
$n$-converter is called recursively on these sub-polynomials. The process stops
when all the sub-polynomials are integers. This procedure is terminating as
polynomials have a maximum degree. Example~\ref{ex:monomial} illustrates the
construction of a term from a
polynomial. Example~\ref{ex:termmul}~\&~\ref{ex:termcnt} construct the
$n$-multiplier and $n$-counter of the resulting term, respectively.
\begin{figure*}
  \centering
  \begin{tikzpicture}[ level 1/.style={level distance = 1.5cm, sibling distance
      = 1.5cm}, level 2/.style={level distance = 1.5cm, sibling distance =
      1.8cm}, level 3/.style={level distance = 1.5cm, sibling distance = 1cm},
    level 4/.style={level distance = 1.5cm, sibling distance = 1cm}, ]
    \node[xshift=-4cm]{$3\cdot x^3 + xy -2\cdot y^2 - 2$} child {node {$-2$}} child
    {node {$3\cdot x^2 + y$} child {node {$3\cdot x$} child {node {$3$}} } child {node
        {$1$}} } child {node {$-2\cdot y $} child {node {$- 2$}} };

    \node[xshift=0cm]{$g$} child {node {$a,a$}} child {node {$F$} child {node
        {$F$} child {node {$F$} child {node {$a$}} child {node {$a$}} } child
        {node {$a$}} } child {node {$F$} child {node {$a$}} child {node
          {$g(a,a,a)$}} } };

    \node[xshift=6cm,yshift=1.5cm]{} child{ node {$F$} child {node {$F$} child
        {node {$F$} child {node {$g(a,a,a,a)$}} child {node {$a$}} } child {node
          {$g(a,a)$}} } child {node {$F$} child {node {$a$}} child {node {$a$}}
      } edge from parent[draw=none]};

    \draw (-4,-7) node {$p(x,y)$}; \draw (0,-7) node {$\cvtnOp{F,p(x,y)}$};
    \draw (6,-7) node {$\cvtpOp{F,p(x,y)}$}; \draw[dotted] (-5.25,-1.2)
    rectangle (-3,-4.75); \draw[dotted] (-1.4,-2.5) rectangle (1,-6.25);
    \draw[dotted] (3.15,-1.25) rectangle (6.2,-5);

  \end{tikzpicture}
  \caption{Applying Definition~\ref{def:convert} to the polynomial of
    Example~\ref{ex:monomial}, we derive $\cvtnOp{F,p(x,y)}$ and
    $\cvtpOp{F,p(x,y)}$. The boxed section of the polynomial tree results in the
    boxed sections of the two term trees. The precise construction is described
    in Example~\ref{ex:monomial}.}
  \label{fig:converterex}
\end{figure*}
\begin{example}\label{ex:monomial}
  Consider the polynomial $p(x,y) = 3\cdot x^3 + xy -2\cdot y^2 - 2$. The positive and
  negative terms representing this polynomial are as follows (See
  Figure~\ref{fig:converterex} for a tree representation):
  \begin{center}
    \textbf{ Positive $n$-Converter}
  \end{center}
  \begin{align*}
    \cvtpOp{F,3\cdot x^3 + xy -2\cdot y^2 - 2} \ \textbf{=}\ &   F(\cvtpOp{F,3\cdot x^2 + y},\cvtpOp{F,-2\cdot y})\\ \vspace*{.5em}
    \cvtpOp{F,3\cdot x^2 + y} \ \textbf{=}\ &  
                                          F(\cvtpOp{F,3\cdot x},\cvtpOp{F,1})\\ \vspace*{.5em}
    \cvtpOp{F,3\cdot x}  \ \textbf{=}\ &  F(\cvtpOp{F,3},a)\\ \vspace*{.5em}
    \cvtpOp{F,3}  \ \textbf{=}\ & g(a,a,a,a)\\ \vspace*{.5em}
    \cvtpOp{F,1}  \ \textbf{=}\ &  g(a,a)\\ \vspace*{.5em}
    \cvtpOp{F,-2\cdot y} \ \textbf{=}\ & F(a,a) \vspace*{.5em}
  \end{align*}
       
    \begin{center}
      \textbf{ Negative $n$-Converter}
    \end{center}
    \begin{align*}
      \cvtnOp{F,3\cdot x^3 + xy -2\cdot y^2 - 2} \ \textbf{=}\ &
                                                         g(a,a,F(\cvtnOp{F,3\cdot x^2 + y},\cvtnOp{F,-2\cdot y}) )\\ \vspace*{.5em}
      \cvtnOp{F,3\cdot x^2 + y}\ \textbf{=}\ & F(\cvtnOp{F,3\cdot x},\cvtnOp{F,1})\\ \vspace*{.5em}
      \cvtnOp{F,3\cdot x} \ \textbf{=}\ & F(a,a)\\ \vspace*{.5em}
      \cvtnOp{F,1} \ \textbf{=}\ & a\\ \vspace*{.5em}
      \cvtnOp{F,-2\cdot y} \ \textbf{=}\ & F(a,g(a,a,a)) \vspace*{.5em}
    \end{align*}
    The result is the following two terms:
    \begin{align*}
      \cvtpOp{F,3\cdot x^3 + xy -2\cdot y^2 - 2}\ \textbf{=}\ &  F(F(F(g(a,a,a,a),a),g(a,a)),F(a,a))\\
      \cvtnOp{F,3\cdot x^3 + xy -2\cdot y^2 - 2} \ \textbf{=}\ & g(a,a,F(F(F(a,a),a),F(a,g(a,a,a)))
    \end{align*}
  \end{example}

  \begin{example}
    \label{ex:termmul}
    Consider the terms constructed in Example~\ref{ex:monomial}, that is, let
    $s=\cvtpOp{F,3\cdot x^3 + xy -2\cdot y^2 - 2}$ and
    $t=\cvtnOp{F,3\cdot x^3 + xy -2\cdot y^2 - 2}$. The $n$-multiplier is computed as
    follows:
    \begin{center}
      \textbf{ Positive $n$-multiplier}
    \end{center}
    \begin{align*}
      \mulOp{F,x,y,s}  \ \textbf{=}\ & 1+x\cdot \mulOp{F,x,y,\cvtpOp{F,3\cdot x^2 + y}} +\\ &
                                                                                      y\cdot \mulOp{F,x,y,\cvtpOp{F,-2\cdot y}} \\ \vspace*{.5em}
      \mulOp{F,x,y,\cvtpOp{F,3\cdot x^2 + y}}  \ \textbf{=}\ & 
                                                           1+ x\cdot\mulOp{F,x,y,\cvtpOp{F,3\cdot x}}\\ \vspace*{.5em}
      \mulOp{F,x,y,\cvtpOp{F,3\cdot x}}  \ \textbf{=}\ &  1\\ \vspace*{.5em}
      \mulOp{F,x,y,\cvtpOp{F,-2\cdot y}} \ \textbf{=}\ & 1 \vspace*{.5em}
    \end{align*}
    \begin{center}
      \textbf{ Negative $n$-multiplier}
    \end{center}
    \begin{align*}
      \mulOp{F,x,y,t}  \ \textbf{=}\ & 1+x\cdot \mulOp{F,x,y,\cdot\cvtnOp{F,3\cdot x^2 + y}} + \\
                                     & y\cdot \mulOp{F,x,y,\cvtnOp{F,-2\cdot y}}\\
      \mulOp{F,x,y,\cvtnOp{F,3\cdot x^2 + y}}  \ \textbf{=}\ & 1+ x\cdot\mulOp{F,x,y,\cvtnOp{F,3\cdot x}} \\
      \mulOp{F,x,y,\cvtnOp{F,3\cdot x}}  \ \textbf{=}\ & 1\\
      \mulOp{F,x,y,\cvtnOp{F,-2\cdot y}} \ \textbf{=}\ & 1 \vspace*{.5em}
    \end{align*}
    \noindent Thus, $\mulOp{F,x,y,s} = \mulOp{F,x,y,t} = 1+x+x^2+y$
  \end{example}
  \begin{example}
    \label{ex:termcnt}
    Consider the terms constructed in Example~\ref{ex:monomial}. The $n$-counter
    is computed as follows:
    \begin{center}
      \textbf{ Positive $n$-counter}
    \end{center}
    \begin{align*}
      \cntOp{F,x,y,a,s}  \ \textbf{=}\ & x\cdot \cntOp{F,x,y,a,\cvtpOp{F,3\cdot x^2 + y}} + \\ &  y\cdot \cntOp{F,x,y,a,\cvtpOp{F,-2\cdot y}}\\
      \cntOp{F,x,y,a,\cvtpOp{F,3\cdot x^2 + y}}  \ \textbf{=}\ & x\cdot \cntOp{F,x,y,a,\cvtpOp{F,3\cdot x}}+ 2\cdot y\\
      \cntOp{F,x,y,a,\cvtpOp{F,3\cdot x}}  \ \textbf{=}\ & x\cdot \cntOp{F,x,y,a,\cvtpOp{F,3}}+y\\
      \cntOp{F,x,y,a,\cvtpOp{F,3}} \ \textbf{=}\ & 4\\
      \cntOp{F,x,y,a,\cvtpOp{F,-2\cdot y}}  \ \textbf{=}\ & x+y\\
    \end{align*}
    \begin{center}
      \textbf{Negative $n$-counter}
    \end{center}
    \begin{align*}
      \cntOp{F,x,y,a,t} \ \textbf{=}\ & 2+ x\cdot \cntOp{F,x,y,a,\cvtnOp{F,3\cdot x^2 + y}}+\\ & y\cdot \cntOp{F,x,y,a,\cvtnOp{F,-2\cdot y}}\\ \vspace*{.5em}
      \cntOp{F,x,y,a,\cvtnOp{F,3\cdot x^2 + y}} \ \textbf{=}\ & x\cdot \cntOp{F,x,y,a,\cvtnOp{F,3\cdot x}}+\\ & y\\\vspace*{.5em}
      \cntOp{F,x,y,a,\cvtnOp{F,3\cdot x}} \ \textbf{=}\ & x+y\\ \vspace*{.5em}
      \cntOp{F,x,y,a,\cvtnOp{F,-2\cdot y}} \ \textbf{=}\  &
                                                        x +y\cdot\cntOp{F,x,y,a,\cvtnOp{F,-2}}\\ \vspace*{.5em}
      \cntOp{F,x,y,a,\cvtnOp{F,-2}} \ \textbf{=}\ & 3 \vspace*{.5em}
    \end{align*}
    \noindent Thus,
    \begin{align*}
      p(x,y)\ \textbf{=}\ & \cntOp{F,x,y,a,s}\ \textbf{=}\  4\cdot x^3 +x^2 y + 3\cdot xy +y^2\\ \vspace*{.5em} 
      q(x,y)\ \textbf{=}\ &  \cntOp{F,x,y,a,t}\ \textbf{=}\  x^3 + x^2 y + 2 \cdot xy + 3\cdot y^2 + 2\\ \vspace*{.5em} 
      p(x,y)-q(x,y) \ \textbf{=}\ & 3x^3+xy-2\cdot y^2-2 \vspace*{.5em} 
    \end{align*}
  \end{example}
  Using the operator introduced in Definition~\ref{def:convert}, we can
  transform a polynomial with integer coefficients into an ASOGU
  problem. The next definition describes the process:
  \begin{definition}
    \label{def:polyUnif}
    Let $p(\overline{x_n})$ be a polynomial and $F\in \Varsf^n$. Then
    $(\mathcal{U},\mathcal{A},F)$ is the ASOGU problem induced by
    $p(\overline{x_n})$ where
    $\mathcal{U}=\{ \cvtnOp{F,p(\overline{x_n})} \unif_F \cvtpOp{F,p(\overline{x_n})}\}$.
  \end{definition}
  \begin{example}
    For the polynomial presented in Example~\ref{ex:monomial}, we can build the
    unification problem $\mathcal{U}=\{s\unif t\}$ where
    \begin{align*}
      s\ \textbf{=}\ &  F\bbbigpl F\bbigpl F\bigpl g(a,a,a,a),a\bigpr ,g(a,a)\bbigpr,F(a,a)\bbbigpr\\
      t\ \textbf{=}\ & g\bbbigpl a,a,F\bbigpl F\bigpl F(a,a),a\bigpr,F\bigpl a,g(a,a,a)\bigpr\bbigpr\bbbigpr
    \end{align*}
    Observe that the $n$-converter will always produce a flex-rigid unification
    equation as long as the input polynomial is of the form
    $p(\overline{x_n})= p'(\overline{x_n})+c$ where $c\not = 0$. When $c=0$, we
    get a flex-flex unification equation, and there is always a solution for
    both the polynomial and the unification equation.
  \end{example}
  The result of this translation is that the $n$-counter captures the structure
  of the polynomial, and the $n$-multipliers cancel out.
  \begin{theorem}
      \label{prop:muleqpoly}
    Let $n\geq 1$, $p(\overline{x_n})$ be a polynomial, and $(\mathcal{U},$
    $\mathcal{A},F)$ an ASOGU problem induced by $p(\overline{x_n})$ where
    $\mathcal{U}=\{ \cvtnOp{F,p(\overline{x_n})} \unif_F
    \cvtpOp{F,p(\overline{x_n})}\}$. Then
    \begin{gather*}
      p(\overline{x_n}) \ \textbf{=}\ \cntrOp{F,\overline{x_n},a,\mathcal{U}} -
      \cntlOp{F,\overline{x_n},a,\mathcal{U}}
      \mbox{ \quad and \quad }
      0 \ \textbf{=}\ \mullOp{F,\overline{x_n},\mathcal{U}}
      -\mulrOp{F,\overline{x_n},\mathcal{U}}.
    \end{gather*}
  \end{theorem}
  \begin{proof}
    We proceed by induction on $\mathit{deg}(p(\overline{x_n}))$.
    \proofcase{Base case:} When $\mathit{deg}(p(\overline{x_n}))=0$, then
    $p(\overline{x_n})= c$ for $c\in \mathbb{Z}$. We have the following cases:
    \begin{itemize}
    \item if $c=0$, then $\mathcal{U}=\{a \unif_F a\}$. This implies that
      $\cntOp{F,\overline{x_n},a,a}= \cntOp{F,\overline{x_n},a,a} = 1$ and 
      $\mulOp{F,\overline{x_n},a}= \mulOp{F,\overline{x_n},a}=0$, i.e., the counter and multiplier of the left and right side are equivalent,
      \vspace*{.5em}
    \item if $c>0$, then $\mathcal{U}=\{s \unif_F t\}$ where
      $\mathit{occ}_{\Sigma}(a,t)$ $=c+1$ and $\mathit{occ}_{\Sigma}(a,s) = 1$. This
      implies that
      $\cntOp{F,\overline{x_n},a,t}-\cntOp{F,\overline{x_n},a,s}= c$ and
      $\mulOp{F,\overline{x_n},s}= \mulOp{F,\overline{x_n},t}=0$, i.e., the counter of the left and right side cancel out to $c$ and the multipliers are equivalent. 
      \vspace*{.5em}
    \item if $c<0$, then $\mathcal{U}=\{s \unif_F t\}$ where
      $\mathit{occ}_{\Sigma}(a,t) = 1$ and $\mathit{occ}_{\Sigma}(a,s) = c+1$. This
      implies that
      $\cntOp{F,\overline{x_n},a,t}-\cntOp{F,\overline{x_n},a,s} = c$ and
      $\mulOp{F,\overline{x_n},s}= \mulOp{F,\overline{x_n},t}=0$, i.e., the counter of the left and right side cancel out to $c$ and the multipliers are equivalent. 
    \end{itemize}
    These arguments complete the base case.%
    \proofcase{Step case:} Let $\mathcal{U}=\{s \unif_F t\}$. We assume for our induction hypothesis that for all
    polynomials $p(\overline{x_n})$ with $\mathit{deg}(p(\overline{x_n}))\leq k$
    the statement holds, and show the statement holds for polynomials
    $p(\overline{x_n})$ with $\mathit{deg}(p(\overline{x_n}))=k+1$. observe that
    \begin{align*}
      \cntOp{F,\overline{x_n},a,t}\ \textbf{=}\ & \cntOp{F,\overline{x_n},a,  \cvtpOp{F,p(\overline{x_n})}} \ \textbf{=}\ \cntOp{F,\overline{x_n},a,  \cvtpOp{F,p(\overline{x_n})_0}}  +\\
                                                 &\sum_{i=1}^n x_i\cdot \cntOp{F,\overline{x_n},a,  \cvtpOp{F,p(\overline{x_n})_{x_i}\downarrow}} \\
      \cntOp{F,\overline{x_n},a,s}\ \textbf{=}\ & \cntlOp{F,\overline{x_n},a,  \cvtnOp{F,p(\overline{x_n})}} \ \textbf{=}\
                                                   \cntOp{F,\overline{x_n},a,  \cvtnOp{F,p(\overline{x_n})_0}}  +\\
                                                 &\sum_{i=1}^n x_i\cdot \cntOp{F,\overline{x_n},a,  \cvtnOp{F,p(\overline{x_n})_{x_i}\downarrow}}
    \end{align*}
    Observe that $a$ is an arity zero constant, thus we do not increment the sum during the recursion (See Definition~\ref{def:ncounter}, bullets 3 and 4). Furthermore, when $p(\overline{x_n})_0=0$, both $\cntOp{F,\overline{x_n},a,  \cvtpOp{F,p(\overline{x_n})_0}}$ and  $\cntOp{F,\overline{x_n},a,  \cvtnOp{F,p(\overline{x_n})_0}}$ are spuriously equivalent to $1$ rather then $0$. However, the terms cancel and do not influence the results).  By the induction hypothesis, we know that for all $0\leq i\leq n$,
    \begin{align*}
      p(\overline{x_n})_{x_i} & \ \textbf{=}\ \cntOp{F,\overline{x_n},a,
                                \cvtpOp{F,p(\overline{x_n})_{x_i}\downarrow}} - \cntOp{F,\overline{x_n},a,  \cvtnOp{F,p(\overline{x_n})_{x_i}\downarrow}}
    \end{align*}
    Thus, we can derive the following:
    \begin{gather*}
      \cntOp{F,\overline{x_n},a,t}-\cntOp{F,\overline{x_n},a,s}  \ \textbf{=}\ p(\overline{x_n})_0 +\sum_{i=1}^n x_i\cdot p(\overline{x_n})_{x_i} = p(\overline{x_n})
    \end{gather*}
    Similarly,
    \begin{align*}
      \mulOp{F,\overline{x_n},t}\ \textbf{=}\ &\mulOp{F,\overline{x_n},  \cvtpOp{F,p(\overline{x_n})}} \ \textbf{=}\   \mulOp{F,\overline{x_n},  \cvtpOp{F,p(\overline{x_n})_0}}+\\
                                               &\sum_{i=1}^n x_i\cdot  \mulOp{F,\overline{x_n},  \cvtpOp{F,p(\overline{x_n})_{x_i}\downarrow}} \\
      \mulOp{F,\overline{x_n},s} \ \textbf{=}\ & \mulOp{F,\overline{x_n},  \cvtnOp{F,p(\overline{x_n})}} \ \textbf{=}\ \mulOp{F,\overline{x_n},  \cvtnOp{F,p(\overline{x_n})_0}}+ \\
                                               &\sum_{i=1}^n x_i\cdot \mulOp{F,\overline{x_n},  \cvtnOp{F,p(\overline{x_n})_{x_i}\downarrow}}
    \end{align*}
    By the induction hypothesis, we know that for all $0\leq i\leq n$,
    \begin{align*}
      0 =& \mulOp{F,\overline{x_n},
           \cvtpOp{F,p(\overline{x_n})_{x_i}\downarrow}}  - \mulOp{F,\overline{x_n}, \cvtnOp{F,p(\overline{x_n})_{x_i}\downarrow}}
    \end{align*}
    Thus, we can derive the following:
    \begin{align*}
        \mulOp{F,\overline{x_n},t}-\mulOp{F,\overline{x_n},s}= & 0 +\sum_{i=1}^n 0\cdot x_i=0 \tag*{\qedhere}
    \end{align*}
  \end{proof}
  A simple corollary of Theorem~\ref{prop:muleqpoly} concerns commutativity of
  unification equations:
  \begin{corollary}
    \label{prop:commute}
    Let $n\geq 1$, $p(\overline{x_n})$ be a polynomial, and
    $(\{s\unif_F t\},\mathcal{A},F)$ an ASOGU problem induced by
    $p(\overline{x_n})$. Then
    $-p(\overline{x_n}) =\cntOp{F,\overline{x_n},a,s} -
    \cntOp{F,\overline{x_n},a,t}.$
  \end{corollary}
  \begin{proof}
    Same as Theorem~\ref{prop:muleqpoly} but swapping terms.
  \end{proof}
  Both $p(\overline{x_n})$ and $-p(\overline{x_n})$ have the same roots, and the induced unification problem is flex-rigid (unless the polynomial is a constant); thus, the induced unification problem uniquely captures the
  polynomial $p(\overline{x_n})$.

  We now prove that the unification condition introduced in
  Lemma~\ref{lem:MainEq} is equivalent to finding the solutions to polynomial
  equations. The following shows how a solution to a polynomial can be obtained
  from the unification condition and vice versa.
  \begin{lemma}\label{lem:reduction-sound-complete}
    Let $p(\overline{x_n})$ be a polynomial, $s = \cvtpOp{F,p(\overline{x_n})}$,
    $t = \cvtnOp{F,p(\overline{x_n})}$ and $(\{ s\unif_F t \},\mathcal{A},F)$ the
    ASOGU problem induced by $p(\overline{x_n})$
    (Definition~\ref{def:polyUnif}). Then there exists non-negative integers
    $h_1,\ldots,h_n\geq 0$ such that
    \begin{gather}\label{eq:sub}
      \sigma =\left\lbrace F\mapsto \lambda \overline{x_n}.g(\overbrace{x_1,\ldots,x_1}^{h_1},\ldots,\overbrace{x_n,\ldots,x_n}^{h_n})\right\rbrace
    \end{gather}
    is a unifier of $\{s\unif_F t\}$ if and only if
    $\{x_i\mapsto h_i\mid 1\leq i\leq n\}$ is a solution to $p(\overline{x_n})=0$.
  \end{lemma}
  \begin{proof}
    Observe, by Theorem~\ref{prop:muleqpoly},
    $\mulOp{F,\overline{h_n},s}
    -\mulOp{F,\overline{h_n},t} = 0$, thus we can ignore the
    $n$-multiplier. We can prove the two directions as follows:
    \proofcase{$\Longrightarrow:$} If $\sigma$ unifies $\{s\unif_F t\}$, then
    $\cntOp{F,\overline{h_n},a,s} =
    \cntOp{F,\overline{h_n},a,t}$ and by Lemma~\ref{lem:MainEq},
    $0 = \cntOp{F,\overline{h_n},a,t}
    -\cntOp{F,\overline{h_n},a,s}.$ By Theorem~\ref{prop:muleqpoly},
    we know that
    $p(\overline{h_n}) = \cntOp{F,\overline{h_n},a,t}
    -\cntOp{F,\overline{h_n},a,s}$. Thus, we derive
    $p(\overline{h_n}) =0$ by transitivity.  \proofcase{$\Longleftarrow:$} If
    $\{x_i\mapsto h_i\mid 1\leq i\leq n\}$ is a solution to
    $p(\overline{x_n}) =0$, we can derive, using Theorem~\ref{prop:muleqpoly},
    that
    $0= \cntOp{F,\overline{h_n},a,t}
    -\cntOp{F,\overline{h_n},a,s}$. Furthermore, we derive that
    $\cntOp{F,\overline{h_n},a,s} =
    \cntOp{F,\overline{h_n},a,t}$. Now, let $\sigma$ be the
    substitution defined in Equation~\ref{eq:sub}. It then follows from the
    definition of the $n$-counter (Definition~\ref{def:ncounter}), that
    $\mathit{occ}_{\Sigma}(a, s\sigma)= \mathit{occ}_{\Sigma}(a, t\sigma)$. Given that
    $g$ is the only other symbol occurring in $s\sigma$ and $t\sigma$ and $g$ is
    associative, it follows that $s\sigma\approx_{\mathcal{A}}t\sigma$. Thus,
    $\sigma$ is a unifier of $\{s\unif_F t\}$.
  \end{proof}
  We have proven that there is a reduction from Hilbert's 10th problem over the
  non-negative integers to ASOGU. We can now state the main result of this
  paper.
  \begin{theorem}\label{thm:sogu-undecidable}
    There exists $n\geq 1$ such that $n$-ASOGU is undecidable.
  \end{theorem}
  \begin{proof}
    The statement follows from the reduction presented in
    Lemma~\ref{lem:reduction-sound-complete} and Theorem~\ref{thm:hilbert-10}.
  \end{proof}
  The following example illustrates this encoding.
  \begin{example}
    Consider the following polynomial:
$$p(x,y) = (x-1)(x-2)(y-1)(y-2)$$
Expansion and reduction result in the following polynomial:
$$ x^2y^2 - 3x^2y - 3xy^2 + 9xy + 2x^2  + 2y^2 - 6y -6x + 4 $$
We encode this polynomial as the following \textit{unification problem} where
$\textbf{F}$ is a second-order variable, $a$ is a constant, $g$ is an
associative binary function symbol:

\begin{gather*}
{
  \textbf{F}\bbbbigpl g\bbbigpl a^6,\textbf{F}\bbigpl \textbf{F}(a,g(a^3,\textbf{F}(a,a))\bbigpr,\textbf{F}(a,g(a^4)))\bbbigpr,g(a^6,\textbf{F}(a,a))\bbbbigpr}\\
\unif\\
{g\bbbbigpl a^4,\textbf{F}\bbbigpl \textbf{F}\bbigpl g\bigpl a^2,\textbf{F}(a,\textbf{F}(a,g(a^2)))\bigpr,g(a^9,\textbf{F}(a,a))\bbigpr,\textbf{F}(a,g(a^3))\bbbigpr\bbbbigpr}
\end{gather*}
where $a^n$ abbreviates a sequence of $n$ occurrences of $a$, and terms are
written in \textit{flattened} form, i.e., intermediary occurrences of
associative function symbols are dropped. Observe that the arity of $\textbf{F}$
is equivalent to the number of unknowns in $p(x,y)$.  We can derive solutions
for this unification problem from the zeros of $p(x,y)$:
\begin{itemize}
\item $p(1,0)=0$. The substitution $ \{F\mapsto \lambda x,y.x\}$ unifies the above terms
  resulting in $g(a^7)\unif g(a^7)$.
\item $p(0,2)=0$. The substitution $ \{F\mapsto \lambda x,y.g(y,y)\}$ unifies the above
  terms resulting in $g(a^{16})\unif g(a^{16})$.
\item $p(1,2)=0$. The substitution $ \{F\mapsto \lambda x,y.g(x,y,y)\}$ unifies the above
  terms resulting in $g(a^{55})\unif g(a^{55})$.
\item $p(2,1)=0$. The substitution $ \{F\mapsto \lambda x,y.g(x,x,y)\}$ unifies the above
  terms resulting in $g(a^{65})\unif g(a^{65})$.
\end{itemize}
As described above, each unknown is associated with an argument of \textbf{F}.
\end{example}

Theorem~\ref{thm:sogu-undecidable} partially answers the question posed in
Section~\ref{sect:introduction} by demonstrating that occurrences of first-order
variables within the unification problem do not impact the decidability of
second-order unification over an associative theory. Furthermore, this holds
when only \emph{one} second-order variable is present.  However, it remains open
whether SOGU is decidable over the empty theory. Compared to the work of Otto, Narendran, and Dougherty~\cite{OTTO19981}, we show that the undecidability of second-order
E-unification is not contingent on the equational theory being a length-reducing
rewrite system. As a final point, it is known that the construction used to demonstrate the undecidability of Diophantine equation solving over $\mathbb{N}$ requires a polynomial with at least $9$ unknowns~\cite{10.5555/164759}. This result extends to our work, implying that the arity of the function
variable $F$ needs to be at least $9$. Thus, the decidability of ASOGU remains
open for function variables of lower arity.

As mentioned earlier, Theorem~\ref{thm:sogu-undecidable} holds when
associativity is replaced with \textit{power associativity}. Furthermore, it is
possible to adjust Definition~\ref{def:convert} such that
Theorem~\ref{prop:muleqpoly} and Corollary~\ref{prop:commute} hold for
$n$-SOGU. However, for the sake of clarity, we omitted these details in the
presentation of our results. Observe that associativity is only required for the \textit{only if} part of
Lemma~\ref{lem:reduction-sound-complete}.

\section{Conclusion and Future Work}
\label{sec:conclusion}
We show that associative second-order ground unification is undecidable using a
reduction from Hilbert's 10$^{th}$ problem over non-negative integers. The reduction required two
novel occurrence counting functions related to the \emph{multiplicity operator}~\cite{DBLP:journals/jsyml/Vrijer85,DBLP:conf/csl/AccattoliK10} and their relationship to the existence of a
unifier. Using these operators, we show how to encode a polynomial in reduced
form in a single unification equation that only contains occurrences of an
associative binary function symbol and a single constant. The reduction holds
even in the case when the binary function symbol is interpreted as \emph{power
  associative} only, i.e., $f(x,f(x,x))= f(f(x,x),x)$. It remains open whether
second-order ground unification (the non-associative case) is decidable. We plan
to investigate how our encoding can be used to discover decidable fragments of
both second-order ground unification and its equational extension. Additionally,
we plan to investigate variants of the presented encoding, e.g., other ways to
encode polynomials in unification equations and their associated unification
problems.

\bibliographystyle{alphaurl}
\bibliography{IEEEabrv,easychair}

\end{document}